\pdfoutput=1
\documentclass[10pt,pra,aps,superscriptaddress,nofootinbib,twocolumn]{revtex4-1}
\usepackage{amsmath,bbm}
\usepackage{amsthm}
\usepackage{amsmath}
\usepackage{latexsym}
\usepackage{amssymb}
\usepackage{float}
\usepackage{graphicx}        
\usepackage{color}
\usepackage{xcolor}
\usepackage{mathpazo}
\usepackage{comment}
\usepackage{enumitem}
\usepackage{multirow}
\usepackage{setspace}
\usepackage{graphicx}
\usepackage{caption}
\usepackage{subcaption}

\usepackage[colorlinks=true,linkcolor=blue,citecolor=magenta,urlcolor=blue]{hyperref}
\usepackage{svg}

\newcommand{\be}{\begin{equation}}
\newcommand{\ee}{\end{equation}}
\newcommand{\bea}{\begin{eqnarray}}
\newcommand{\eea}{\end{eqnarray}}

\def\squareforqed{\hbox{\rlap{$\sqcap$}$\sqcup$}}
\def\qed{\ifmmode\squareforqed\else{\unskip\nobreak\hfil
\penalty50\hskip1em\null\nobreak\hfil\squareforqed
\parfillskip=0pt\finalhyphendemerits=0\endgraf}\fi}
\def\endenv{\ifmmode\;\else{\unskip\nobreak\hfil
\penalty50\hskip1em\null\nobreak\hfil\;
\parfillskip=0pt\finalhyphendemerits=0\endgraf}\fi}

\newcommand{\tr}{\text{Tr}}

\newcommand{\ket}[1]{|#1\rangle}
\newcommand{\bra}[1]{\langle#1|}

\makeatletter
\newtheorem*{rep@theorem}{\rep@title}
\newcommand{\newreptheorem}[2]{%
\newenvironment{rep#1}[1]{%
 \def\rep@title{#2 \ref{##1}}%
 \begin{rep@theorem}}%
 {\end{rep@theorem}}}
\makeatother

\newtheorem{thm}{Theorem}
\newreptheorem{thm}{Theorem}
\newtheorem{lemma}{Lemma}

\newtheorem{obs}{Observation}

\newtheorem{proposition}{Proposition}

\newtheorem{sketch of proof}{Sketch of proof}
\begin{document}


\title{No-go theorem for maximally $\psi-$epistemic models using two pure states of a qubit}


\author{Sagnik Ray}
\affiliation{School of Physics, Indian Institute of Science Education and Research Thiruvananthapuram, Kerala 695551, India}

\author{Anubhav Chaturvedi}

\affiliation{Institute of Theoretical Physics and Astrophysics, Faculty of Mathematics, Physics and Informatics, University of Gda\'nsk, 80-308 Gda\'nsk, Poland}
\author{Debashis Saha}
\affiliation{Department of Physics, School of Basic Sciences, Indian Institute of Technology Bhubaneswar, Odisha 752051, India}


\begin{abstract}
{This work addresses the fundamental question concerning the nature of quantum states, specifically to what extent they represent an agent's knowledge about the system, that is, is \textit{epistemic} in nature. We show that, despite the existence of the Kochen-Specker model, which is a maximally $\psi-$epistemic model for single-qubit systems when projective measurements are employed, \textit{no} maximally $\psi-$epistemic model can reproduce statistics arising out of non-projective measurements on a pair of single-qubit states. This has been done in an experimentally robust way using only two pure single-qubit states and their equal mixture, by constructing a two-input, three-output game called \textit{Risk Aware Distinguishability} (RAD) and generalizing the notion of epistemic and quantum overlaps using RAD. Furthermore, we maximize the difference of generalized overlaps instead of minimizing their ratio, and in the process point out why that is a stricter criterion for refuting $\psi-$epistemic models.}

\end{abstract}

\maketitle

\section{Introduction}
 It has been a century since the formalization of quantum theory, and still some of its foundational aspects have remained contentious. The answer to what the quantum state means has been debated since the 1927 Solvay Conference \cite{solvay,epr,Fine2020EPR}. Barring relational \cite{Rovelli2025Relational}, operationalist, and other Copenhagenish \cite{schmid25} approaches to quantum theory, one must contend with the fact that any system must be described by a set of attributes which are intrinsic to the system and observer-independent, called the \textit{ontic state}. This notion, formalized in the ontological models framework \cite{Harrigan_2010} has been used to probe whether quantum states uniquely encode these ontic states, that is, if the quantum states are \textit{ontic} in nature. While Pusey, Barrett, and Rudolph demonstrated in \cite{PBR} that they are, they did so in the context of composite systems by assuming the ontic states of independently prepared systems to be uncorrelated. Explicit examples where quantum states are treated as mere representations of knowledge about the system \cite{Lewis,Aaronson}, that is, are \textit{epistemic} in nature, exist when one discards such auxiliary assumptions as made in \cite{PBR}. It is possible, however, to heavily constrain such models where the quantum state is epistemic--- called $\psi-$epistemic models, without any auxiliary assumptions, as has been shown in \cite{Barrett,Leifer,Branciard}. The question we are interested in, is not whether one can constrain such $\psi-$epistemic models, which has already been demonstrated in the aforementioned works. Rather, we are interested in the minimum requirements in terms of the dimension of the quantum preparations and the number of such preparations to constrain such models meaningfully.


In this article, we specifically focus on single-qubit systems. We do so because, firstly, they are the lowest-dimensional non-trivial system, and secondly because they have, so far, been left untouched in these prior works due to the existence of a \textit{maximally} $\psi-$epistemic model proposed by Kochen and Specker \cite{KS}, albeit restricted to projective measurements. A maximally $\psi-$epistemic model is one where the \textit{epistemic overlap}---which is the overlap between two distributions corresponding to two quantum states over the ontic state space--- completely accounts for the \textit{quantum overlap}--- which is a measure of operational indistinguishability of those two quantum states. This eliminates the need for quantum theory to calculate their operational indistinguishability. Since notions of epistemic and quantum overlaps which are derived from operational distinguishability have been unable to constrain $\psi-$epistemic models for a qubit, we need a quantum task with two inputs but at least three outputs instead of two (as is the case with distinguishability) to try to constrain such models, and that is what has been done in this work, by considering a task we call \textit{Risk Aware Distinguishability} (RAD). 

 We generalize the notion of maximal $\psi-$epistemicity by formulating it in a decision-theoretic way, where the epistemic and quantum overlaps are defined with respect to an arbitrary binary decision task, instead of defining it solely with respect to distinguishability. That allows us to formulate a \textit{universal} notion of maximal $\psi-$epistemicity, where we require the overlaps to match for every possible binary decision task. Then, by choosing a task such as RAD, we show that \textit{no} universally maximal $\psi-$epistemic model exists for single-qubit states. Thus,  unlike the Kochen-Specker model, which was able to explain quantum statistics emerging out of projective measurements on a qubit \cite{cs}, no such maximally $\psi-$epistemic model exists that can explain statistics when non-projective Positive Operator Valued Measurements (POVMs) are employed. What is remarkable is the fact that just \textit{two} pure states of a qubit are sufficient to show this, which is something that, to our knowledge, has not been achieved before. The gap between the overlaps is bounded by an operational quantity, which means that it is robustly testable in an experimental setting. 

While formulated without any explicit reference to qubit states, the maximal value of the bound, retrieved via a specifically developed tracial noncommuting polynomial optimization scheme \cite{Tavakoli2022informationally,Chaturvedi2021characterising}, is saturated by a pair of them. The minimal requirement so far for constraining $\psi-$epistemic models--- in terms of the number of states and their dimension both turn out to be two.
 Finally, our operational bound also provides a quantum advantage in a constrained communication scenario of the form considered in \cite{Manna_2024,pandit25}. Here the first two messages are pure states, their equal mixture is used as the third message, and the constraint is a bound on their optimal reward in RAD.

In the following sections, we provide a brief recap of ontological models, expand on the notion of maximal $\psi-$epistemicity, introduce RAD as a game being played between our old friends Alice and Bob, and from then on, present the results.
 \section{Ontological Models, Maximal $\psi-$epistemicity and Risk Aware Distinguishability}

\subsection{Ontological models and overlaps of quantum states} Quantum theory is an \textit{operational theory} insofar as its main goal is to predict measurement statistics in a lab setting. It refrains from saying anything about how the measurement statistics emerge from the underlying attributes of the system being measured. In an ontological model \cite{Harrigan_2010}, we call the complete specification of the 'underlying attributes' an \textit{ontic state}, $\lambda$, and it belongs to a measurable set $\Lambda$ referred to as the \textit{ontic state space}. Every quantum state $\ket{\psi}$ corresponds to a probability distribution over $\Lambda$ called the \textit{epistemic state} $\{\mu(\lambda|\psi)\}$ such that $\mu(\lambda|\psi)\geq0$ $\forall \lambda\in\Lambda$ and $\int_{\Lambda}\mu(\lambda|\psi)d\lambda=1$. For every $n-$outcome measurement $\mathcal{M}=\{M_k\}_{k=1}^n$, one can define a \textit{response scheme}, $\{\xi(k|\lambda,\mathcal{M})\}_{k=1}^n$ such that $\xi(k|\lambda,\mathcal{M})\geq0 \quad \forall k$ and $\sum_k \xi(k|\lambda,\mathcal{M})=1 \quad \forall \lambda,\mathcal{M}$. Ultimately, they must reproduce the measurement statistics predicted by quantum theory, such that, 
\begin{equation} \label{ontic_prob}
    p(k|\psi,\mathcal{M})=\bra{\psi}M_k\ket{\psi}=\int_{\Lambda}\mu(\lambda|\psi)\xi(k|\lambda,\mathcal{M})d\lambda.
\end{equation}   

Classification of such ontological models into $\psi-$\textit{ontic}, \textit{maximally} $\psi-$\textit{epistemic} and \textit{non-maximally} $\psi-$\textit{epistemic} is generally done on the basis of how well the overlap between two epistemic states-- called the \textit{epistemic overlap}--of non-orthogonal quantum preparations is in accounting for their operational indistinguishability. Distinguishability of two quantum states $\ket{\psi_1}$ and $\ket{\psi_2}$ with weights $w_1>0$ and $w_2>0$ is defined as,
\begin{equation} \label{weighted_dist}
    D_Q(\{\psi_1,w_1\},\{\psi_2,w_2\})=\max_{\{\mathcal{M}\}} \sum_{i=1}^2 w_ip(i|\psi_i,\mathcal{M}) ,
\end{equation}
where, for the case in which $w_1=w_2=1/2$, we shall simply use the notation $D_Q(\psi_1,\psi_2)$. The amount of indistinguishability is captured by the \textit{quantum overlap},
\begin{equation}
    \omega_Q(\psi_1,\psi_2)=2(1-D_Q(\psi_1,\psi_2)).
\end{equation}
If one can construct an ontological model where the epistemic overlap,
\begin{equation}
   \omega_{\Lambda}(\psi_1,\psi_2)=\int_{\Lambda}\min(\mu(\lambda|\psi_1),\mu(\lambda|\psi_2))d\lambda, 
\end{equation}
is equal to $\omega_Q(\psi_1,\psi_2)$, then that model is called \textit{maximally $\psi-$epistemic}. If $\omega_{\Lambda}(\psi_1,\psi_2)$ falls short of completely accounting for $\omega_{Q}(\psi_1,\psi_2)$ then we shall call that model \textit{non-maximally $\psi-$epistemic}. Finally, if for two non-orthogonal quantum states, $\omega_{\Lambda}(\psi_1,\psi_2)=0$, then we shall call it \textit{$\psi-$ontic}. As mentioned in the introduction, successive works by Barrett et al. \cite{Barrett}, Leifer \cite{Leifer}, and Branciard \cite{Branciard} have claimed that there exist quantum states where $\psi-$epistemic models become asymptotically bad at explaining their quantum overlaps; however, they have all left single-qubit systems untouched in their analyses. This is primarily because of the existence of a maximally $\psi-$epistemic model for pure qubit states with projective measurements, proposed by Kochen and Specker \cite{KS}. Moreover, it is shown that any two-outcome POVM can be simulated using projective measurements \cite{Oszmaniec-prl}. Therefore, to probe the nonexistence of a maximally $\psi$-epistemic model, even for pure qubit states, we must consider an operational entity that involves at least three-outcome POVM and simultaneously connects with the epistemic overlap. That is where RAD becomes important. However, before discussing RAD, we provide a general formulation of maximal $\psi-$epistemicity, which, as we will see, provides a justification for using RAD later on in this work.

\subsection{Formulation of universally maximal $\psi-$epistemicity} \label{decision_theoretic}

The usual notion of maximal $\psi$-epistemicity, as shown in the preceding subsection, is formulated in terms of the
minimum-error discrimination, or the distinguishability of two quantum states. Here we introduce a more
general, decision-theoretic formulation in which distinguishability is viewed
as one particular instance of a broader class of prepare-and-measure tasks.

Maximal $\psi$-epistemicity identifies the quantum overlap with the corresponding epistemic overlap. A meaningful definition of overlap should satisfy two basic requirements. First, it should admit an operational interpretation in terms of a well-defined task. Second, it should be monotonic under physical transformations of the two preparations: under a quantum channel, or its ontic analogue, the corresponding overlap should not decrease.

The first requirement provides an unambiguous prescription for defining both the quantum and epistemic overlaps directly from an operational task. The second expresses the natural intuition that physical processing cannot make two preparations more distinguishable. Consequently, whenever a pair of quantum states, or the corresponding pair of epistemic states, becomes less distinguishable under such a transformation, its overlap should not decrease. In fact, as demonstrated in Appendix \ref{sec:general-maximal-psi-epistemicity}, the first implies the second.

Consider two pure-state preparations $\{\ket{\psi_x}\}_x,  x\in\{1,2\},$
and a measurement $\mathcal{M}=\{M_z\}_{z\in\mathcal Z}$ with outcome $z$. Let
$g(z,x)\in\mathbb{R}$ denote the payoff associated with reporting outcome
$z$ when preparation $x$ is performed. We define the quantum score of the payoff as
\begin{equation}
   S_Q^g(\psi_1,\psi_2)
    :=
    \max_{\{\mathcal{M}\}}  \sum_{x=1}^{2}\sum_{z\in\mathcal Z}
    g(z,x)\,p(z|\psi_x,\mathcal{M}),
    \label{eq:general-linear-score}
\end{equation}
where any choice of prior probabilities may be absorbed into the payoff
function $g$.
Thus, the payoff functional is linear in the operational probabilities.

For a nontrivial task, it is useful to introduce two reference values. If the
states are orthogonal, the preparation label $x$ can be determined perfectly,
and hence
\begin{equation}
    S_\perp^g
    :=
    \sum_{x=1}^{2}\max_z g(z,x).
    \label{eq:Sperp}
\end{equation}
On the other hand, if the two states are identical, their measurement
statistics coincide, and therefore
\begin{equation}
    S_=^g
    :=
    \max_z\left[g(z,1)+g(z,2)\right].
    \label{eq:Sequal}
\end{equation}
We shall restrict attention to nontrivial decision tasks satisfying
$S_\perp^g>S_=^g$.
Suppose now that the operational statistics $\{p(z|\psi_x,\mathcal{M})\}_{z,x}$ admit an ontological
representation as described in \eqref{ontic_prob}.
 If the ontic state $\lambda$ were directly accessible, one could
choose, for every $\lambda$, the action $z$ yielding the largest payoff.
This motivates the classical, or ontic, score
\begin{equation}
    S_{\Lambda}^g(\psi_1,\psi_2)
    :=
    \int_\Lambda d\lambda\,
    \max_z
    \left[
        g(z,1)\mu(\lambda|\psi_1)
        +
        g(z,2)\mu(\lambda|\psi_2)
    \right].
    \label{eq:classical-general-score}
\end{equation}
Therefore,
\begin{equation}
    S_Q^g(\psi_1,\psi_2)
    \leq
    S_{\Lambda}^g(\psi_1,\psi_2).
    \label{eq:quantum-classical-score-ordering}
\end{equation}
We may consequently associate with every such decision task a quantum overlap
and an epistemic overlap according to
\begin{align}
    \omega_Q^g(\psi_1,\psi_2)
    &:=
    \frac{
        S_\perp^g-S_Q^g(\psi_1,\psi_2)
    }{
        S_\perp^g-S_=^g
    },
    \label{eq:general-quantum-overlap}\\
    \omega_{\Lambda}^g(\psi_1,\psi_2)
    &:=
    \frac{
        S_\perp^g-S_{\Lambda}^g(\psi_1,\psi_2)
    }{
        S_\perp^g-S_=^g
    }.
    \label{eq:general-classical-overlap}
\end{align}
With this normalization, orthogonal states have vanishing quantum overlap,
whereas identical states have unit quantum overlap. Moreover, \eqref{eq:quantum-classical-score-ordering} immediately implies
\begin{equation}
    \omega_{\Lambda}^g(\psi_1,\psi_2)
    \leq
    \omega_Q^g(\psi_1,\psi_2).
    \label{eq:general-overlap-ordering}
\end{equation}

This leads naturally to the following definition. An ontological model is
\emph{maximally $\psi$-epistemic with respect to the decision task $g$} if
\begin{equation}
    \omega_{\Lambda}^g(\psi_1,\psi_2)
    =
    \omega_Q^g(\psi_1,\psi_2),
    \ \ 
    \text{for all pure states }\{\psi_1,\psi_2\}.
    \label{eq:maximal-psi-epistemic-task}
\end{equation}
A task-independent notion may then be obtained by requiring \eqref{eq:maximal-psi-epistemic-task} to hold for every admissible binary
decision task $g$. Thus, we define \textit{universally maximal $\psi-$epistemicity for an ontological model as the scenario where the
access to the ontic state $\lambda$ provides no advantage over quantum theory
in \emph{any} binary decision problem of the above form.} \\
This is equivalent to \textit{epistemically complete} ontic model, proposed in \cite{epis_incom_qt}, for all possible pairs of preparations.\\

One can see that $S_Q^g(\psi_1,\psi_2)=D_Q(\psi_1,\psi_2)$ when $g(z,x)=\delta_{z,x}/2$, for $z,x\in\{1,2\}$, and this results in the standard notions of epistemic and quantum overlaps as discussed in the preceding subsection. We shall consider a different set of $\{g(z,x)\}_{z,x}$ that defines the task RAD, and in the process show that even if a pair of single-qubit states is maximally $\psi-$epistemic with respect to distinguishability, it is not so in a task-independent manner.

 \subsection{Risk Aware Distinguishability}
 Risk Aware Distinguishability can be understood as a game of chance that Alice decides to play with Bob. They both walk inside Bob's laboratory, where he has prepared multiple copies of two quantum states, $\ket{\psi_1}$ and $\ket{\psi_2}$. Bob randomly presents one of the states to Alice, and she needs to guess which one it is. If she guesses correctly, she gets a unit reward, and if she guesses wrong, she gets a penalty of value $-\beta$, where $\beta \in [0,1]$. She also has a third option, where, regardless of the state, if she answers '$3$', she receives a reward of value $\alpha$, where $\alpha \in [0,1]$. Alice, ever the physicist, optimizes over every possible three-outcome measurement $\{\mathcal{M}\}$ to calculate the \textit{best average reward} she can win in this game, 
\begin{subequations} \label{S_Q}
    \begin{align}
         S^{RAD}_{Q} =&\frac{1}{2}\max_{\{\mathcal{M}\}}\Big(p(1|\psi_1,\mathcal{M})-\beta p(2|\psi_1,\mathcal{M})+\alpha p(3|\psi_1,\mathcal{M}) \notag \\
       &+ p(2|\psi_2,\mathcal{M})-\beta p(1|\psi_2,\mathcal{M})+\alpha p(3|\psi_2,\mathcal{M})\Big) \label{S_M form 1}\\
       =&\frac{1}{2}\max_{\{\mathcal{M}\}}\Big((1+\beta)\big(p(1|\psi_1,\mathcal{M})+p(2|\psi_2,\mathcal{M})\big) \notag \\
    &+(\alpha+\beta)\big(p(3|\psi_1,\mathcal{M})+p(3|\psi_2,\mathcal{M})\big)\Big)-\beta, \label{S_M form 2}
    \end{align}
\end{subequations}
where \eqref{S_M form 2} follows from \eqref{S_M form 1} since $\sum_{i=1}^3p(i|\psi_j,\mathcal{M})=1$ for $j=1,2$. The game is somewhat similar to distinguishability where Alice needs to correctly guess the quantum state, however, since she incurs a penalty every time she guesses wrong, she needs to be 'aware' of the risk involved, and choose her strategy accordingly. Henceforth in this article, we shall denote the best average reward as $S^{RAD}_Q(\psi_1,\psi_2;\alpha,\beta)$ to make the dependence on the reward and penalty explicit. Note that if $S_Q^{RAD}(\psi_1,\psi_2;\alpha,\beta)$ is expressed according to \eqref{eq:general-linear-score}, then, $g(1,1)=g(2,2)=1/2$, $g(2,1)=g(1,2)=-\beta/2$, and  $g(3,1)=g(3,2)=\alpha/2$. A closed-form expression for $S^{RAD}_Q(\psi_1,\psi_2;\alpha,\beta)$ is not available in general, thus, for a given pair of states and a given value of $\alpha$ and $\beta$, the value of $S^{RAD}_Q(\psi_1,\psi_2;\alpha,\beta)$ is obtained by solving a straightforward SDP (Semi-Definite Program) which satisfies Slater's strong duality condition, implying that it yields the optimal value up to machine precision. This game, it turns out, is valuable for constraining $\psi$-epistemic models, as demonstrated in the remainder of the article.

\subsection{Generalized overlaps of quantum states}
 Taking inspiration from the general formula provided in Section \ref{decision_theoretic}, we generalize the idea of epistemic and quantum overlaps by defining them with respect to the best average reward of RAD, $S^{RAD}_Q(\psi_1,\psi_2;\alpha,\beta)$. Since $S^{RAD}_Q(\psi_1,\psi_2;0,0)=D_Q(\psi_1,\psi_2)$, the generalization of the overlaps is quite natural and it reduces to the standard overlaps for $\alpha=\beta=0$. The following theorem provides us with the definition of generalized overlap.

\begin{thm} \label{def_w_lambda}
    For every value of $\alpha$ and $\beta$ in the range $[0,1]$, and two quantum states $\ket{\psi_1}$ and $\ket{\psi_2}$, the following inequality holds,
     \begin{equation} \label{def_epis_overlap_RAD}
    \begin{split}
        &S^{RAD}_Q(\psi_1,\psi_2;\alpha,\beta)\leq S^{RAD}_{\Lambda}(\psi_1,\psi_2;\alpha,\beta) \textrm{ and,}\\
        &S^{RAD}_{\Lambda}(\psi_1,\psi_2;\alpha,\beta)=1-\frac{\omega_{\Lambda}(\psi_1,\psi_2;\alpha,\beta)}{2}, 
        \end{split}
    \end{equation}
    where $\omega_{\Lambda}(\psi_1,\psi_2;\alpha,\beta)$ is the generalized epistemic overlap of $\ket{\psi_1}$ and $\ket{\psi_2}$ and is defined in terms of three unnormalized distributions, $\{\tilde{\mu}_1(\lambda)\},\{\tilde{\mu}_2(\lambda)\}$ and $\{\tilde{\mu}_3(\lambda)\}$. These are defined as,
    \begin{equation} \label{dist_def}
        \begin{split}
            &\tilde{\mu}_1(\lambda)=(1+\beta)\mu(\lambda|\psi_1),\\
            &\tilde{\mu}_2(\lambda)=(1+\beta)\mu(\lambda|\psi_2),\\
            &\tilde{\mu}_3(\lambda)=(\alpha+\beta)\left(\mu(\lambda|\psi_1)+\mu(\lambda|\psi_2)\right),
        \end{split}
    \end{equation}
    such that,
    \begin{equation} \label{w_S_mod_rep}
\begin{split}
&\omega_{\Lambda}(\psi_1,\psi_2;\alpha,\beta)=T_1+T_2+T_3-T_4 \textrm{ } \textrm{ where,}\\
&T_1=\int_{\Lambda}\min(\tilde{\mu}_1(\lambda),\tilde{\mu}_2(\lambda))d\lambda,\\
&T_2=\int_{\Lambda}\min(\tilde{\mu}_1(\lambda),\tilde{\mu}_3(\lambda))d\lambda-(\alpha+\beta),\\
&T_3=\int_{\Lambda}\min(\tilde{\mu}_2(\lambda),\tilde{\mu}_3(\lambda))d\lambda-(\alpha+\beta),\\
&T_4=\int_{\Lambda}\min(\tilde{\mu}_1(\lambda),\tilde{\mu}_2(\lambda),\tilde{\mu}_3(\lambda))d\lambda.
\end{split}
\end{equation}
\end{thm}

  The proof is straightforward and has been provided in Appendix \ref{appen_proof_th1}. Here  $S^{RAD}_{\Lambda}(\psi_1,\psi_2;\alpha,\beta)$ signifies the best average reward that Alice could have won had she known the ontic state of the system. The generalized epistemic overlap, $\omega_{\Lambda}(\psi_1,\psi_2;\alpha,\beta)$, then represents the degree to which one is unable to win a reward even when they have access to the ontic states. When compared to operational distinguishability, $S^{RAD}_{\Lambda}(\psi_1,\psi_2;\alpha,\beta)$ is analogous to ontic distinctness of the quantum states \cite{cs} and $\omega_{\Lambda}(\psi_1,\psi_2;\alpha,\beta)$ is analogous to the standard epistemic overlap that quantifies the degree to which the states are ontologically indistinct. The generalized epistemic overlap is represented in Figure \ref{fig:gen_overlap}.

\begin{figure}
    \centering
    \begin{subfigure}{.4\textwidth}
  \centering
  \includegraphics[width=1\linewidth]{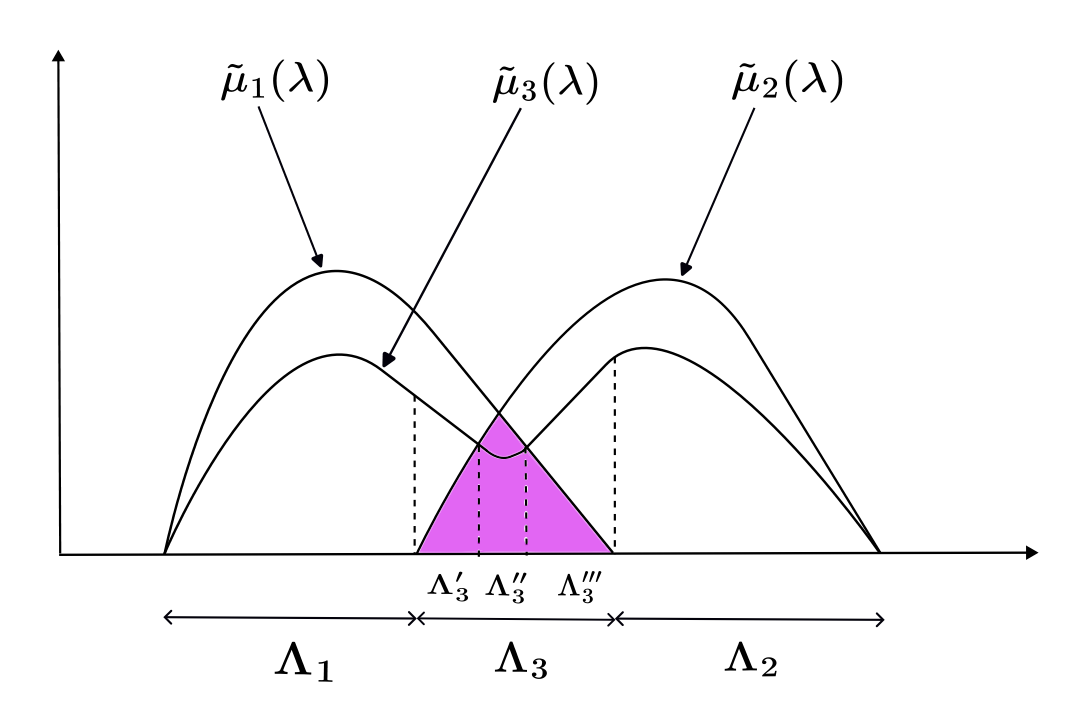}
  \caption{Case 1: $\alpha \leq (1-\beta)/2.$ The area of the pink-shaded region denotes the generalized epistemic overlap.}
  \label{fig:sub1}
\end{subfigure}
\begin{subfigure}{.4\textwidth}
  \centering
  \includegraphics[width=1\linewidth]{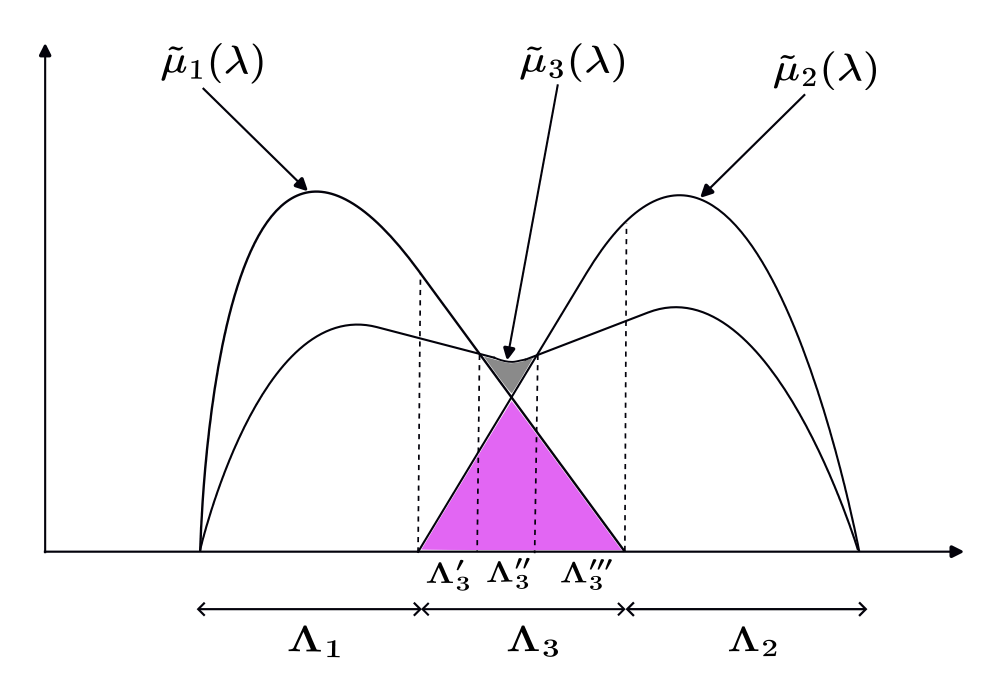}
  \caption{Case 2: $\alpha > (1-\beta)/2.$ The generalized overlap is the net area of the pink-shaded region with the gray-shaded region's area subtracted from it.}
  \label{fig:sub2}
\end{subfigure}
    \caption{Schematic representation of the generalized epistemic overlap of two states $\ket{\psi_1}$ and $\ket{\psi_2}$. The ontic space has been divided into three mutually exclusive regions, $\Lambda_1$, $\Lambda_2$, and $\Lambda_3$, with $\Lambda_3$ being further subdivided into $\Lambda_3'$, $\Lambda_3''$, and $\Lambda_3'''$. $\Lambda_1$ is defined as the region where $\tilde{\mu}_2(\lambda)=0$ and $\tilde{\mu}_1(\lambda)\geq\tilde{\mu}_3(\lambda)$ while $\Lambda_2$ is the region where $\tilde{\mu}_1(\lambda)=0$ and $\tilde{\mu}_2(\lambda)\geq\tilde{\mu}_3(\lambda)$. In the region $\Lambda_3$, all three distributions are non-zero such that in $\Lambda_3'$, $\tilde{\mu}_1(\lambda)\geq\tilde{\mu}_3(\lambda)\geq\tilde{\mu}_2(\lambda)$ while in $\Lambda_3'''$, $\tilde{\mu}_2(\lambda)\geq\tilde{\mu}_3(\lambda)\geq\tilde{\mu}_1(\lambda)$. In $\Lambda_3''$ either $\tilde{\mu}_3(\lambda)\leq \tilde{\mu}_1(\lambda)$ and $\tilde{\mu}_2(\lambda)$ or $\tilde{\mu}_3(\lambda)\geq \tilde{\mu}_1(\lambda)$ and $\tilde{\mu}_2(\lambda)$ depending upon whether $\alpha\leq(1-\beta)/2$ or $\alpha>(1-\beta)/2$.} 
    \label{fig:gen_overlap}
\end{figure}
Notice that the behavior of $\tilde{\mu}_3(\lambda)$ is parameter-dependent such that for $\alpha\leq (1-\beta)/2$, $\tilde{\mu}_3(\lambda)\leq \tilde{\mu}_1(\lambda)$ and $\tilde{\mu}_2(\lambda) \quad \forall\lambda\in\Lambda_3''$ whereas for $\alpha>(1-\beta)/2$, $\tilde{\mu}_3(\lambda)> \tilde{\mu}_1(\lambda)$ and $\tilde{\mu}_2(\lambda) \quad \forall\lambda\in\Lambda_3''$. The reason why $\alpha=(1-\beta)/2$ in the parameter space is the line across which the behavior of $\tilde{\mu}_3(\lambda)$ changes is explained in the next few lines. Its influence on $\tilde{\mu}_3(\lambda)$ leads to a bifurcation in the way the game itself behaves, as we shall see shortly.

Since $\tilde{\mu}_3(\lambda)$ is defined as the sum of $\mu(\lambda|\psi_1)$ and $\mu(\lambda|\psi_2)$ times $(\alpha+\beta)$, for the ontic subspaces $\Lambda_1$ and $\Lambda_2$, it stays strictly lesser than $\tilde{\mu}_1(\lambda)$ and $\tilde{\mu}_2(\lambda)$ respectively. However, in the ontic subspace $\Lambda_3$ which forms the common support of both $\mu(\lambda|\psi_1)$ and $\mu(\lambda|\psi_2)$ and thereby of both $\tilde{\mu}_1(\lambda)$ and $\tilde{\mu}_2(\lambda)$, there exists a region in this subspace-- $\Lambda_3''$, where $\tilde{\mu}_3(\lambda)$ is either lesser than or equal to both $\tilde{\mu}_1(\lambda)$ and $\tilde{\mu}_2(\lambda)$ or greater than both $\tilde{\mu}_1(\lambda)$ and $\tilde{\mu}_2(\lambda)$. If we consider the situation as described in Figure \ref{fig:sub1}, then for all $\lambda \in \Lambda_3''$,
\begin{equation}
    \begin{split}
        &(\alpha+\beta)\big(\mu(\lambda|\psi_1)+\mu(\lambda|\psi_2)\big)\leq (1+\beta)\mu(\lambda|\psi_1) \textrm{ and,}\\
        &(\alpha+\beta)\big(\mu(\lambda|\psi_1)+\mu(\lambda|\psi_2)\big)\leq (1+\beta)\mu(\lambda|\psi_2).
    \end{split}
\end{equation}
One can add the above inequalities, which leads to $\alpha\leq (1-\beta)/2$. A similar calculation for the case described in Figure \ref{fig:sub2} leads to $\alpha>(1-\beta)/2$. Thus, as we cross the line $\alpha=(1-\beta)/2$ in the parameter space, the behavior of $\tilde{\mu}_3(\lambda)$ changes. As a result of this change in behavior of $\tilde{\mu}_3(\lambda)$, the values of $T_2$, $T_3$, and $T_4$ change as well. This leads us to the observation,

\begin{obs} \label{obs_1}
    The generalized epistemic overlap $\omega_\Lambda(\psi_1,\psi_2;\alpha,\beta)$ of two quantum states $\ket{\psi_1}$ and $\ket{\psi_2}$ exhibits two distinct behaviors depending on whether $\alpha\leq (1-\beta)/2$ or $\alpha>(1-\beta)/2$,

\begin{equation} \label{w_S_projective}
        \omega_{\Lambda}(\psi_1,\psi_2;\alpha,\beta)\Bigg|_{\alpha\leq\frac{1-\beta}{2}}=\int_{\Lambda_3}\min(\tilde{\mu}_1(\lambda),\tilde{\mu}_2(\lambda))d\lambda,
\end{equation}
and,
\begin{equation} \label{w_S_POVM}
\begin{split}
    &\omega_{\Lambda}(\psi_1,\psi_2;\alpha,\beta)\Bigg|_{\alpha>\frac{1-\beta}{2}}=\int_{\Lambda_3}\min(\tilde{\mu}_1(\lambda),\tilde{\mu}_2(\lambda))d\lambda\\
    &+\int_{\Lambda_3''}\max(\tilde{\mu}_1(\lambda),\tilde{\mu}_2(\lambda))d\lambda-\int_{\Lambda_3''}\tilde{\mu}_3(\lambda)d\lambda.
\end{split}
\end{equation}
\end{obs}
 The detailed calculation has been presented in Appendix \ref{appen_gen_overlap}. It is clear from \eqref{w_S_projective} that for $\alpha\leq(1-\beta)/2$, $\omega_{\Lambda}(\psi_1,\psi_2;\alpha,\beta)$ is equivalent to $\omega_{\Lambda}(\psi_1,\psi_2;0,0)$ up to a multiplicative constant of $(1+\beta)$, and thus RAD is equivalent to a distinguishability game which is optimized by a two-outcome projective measurement. For qubit states, then, this parameter regime provides us with no new insight. However, for $\alpha>(1-\beta)/2$, $\omega_{\Lambda}(\psi_1,\psi_2;\alpha,\beta)$ is \textit{not reducible} to $\omega_{\Lambda}(\psi_1,\psi_2;0,0)$, and thus the optimization must involve a genuine three-outcome measurement, which for a qubit implies an optimization over non-projective POVMs. 

Similar to the standard notion of quantum overlap, one can define the \textit{generalized quantum overlap} as,
 \begin{equation}
    \omega_Q(\psi_1,\psi_2;\alpha,\beta)=2(1-S^{RAD}_Q(\psi_1,\psi_2;\alpha,\beta)),
\end{equation}
 where $\omega_Q(\psi_1,\psi_2;\alpha,\beta)$ quantifies the degree to which one is unable to win a reward in a game of RAD. Defined in a way analogous to $\omega_{\Lambda}(\psi_1,\psi_2;\alpha,\beta)$ in \eqref{def_epis_overlap_RAD}, it is its operational counterpart and since $S_Q^{RAD}(\psi_1,\psi_2;\alpha,\beta)\leq S_\Lambda^{RAD}(\psi_1,\psi_2;\alpha,\beta)$, it immediately follows that $\omega_{\Lambda}(\psi_1,\psi_2;\alpha,\beta)\leq \omega_{Q}(\psi_1,\psi_2;\alpha,\beta)$.  Thus, now the categorization of ontological models into maximally $\psi-$epistemic, non-maximally $\psi-$epistemic or $\psi-$ontic shall be done by how well $\omega_{\Lambda}(\psi_1,\psi_2;\alpha,\beta)$ is in accounting for $\omega_Q(\psi_1,\psi_2;\alpha,\beta)$. The notions of overlaps discussed in Section \ref{decision_theoretic} are normalized and hence can achieve a maximum value $\omega_{\textrm{max}}=1$. However, the generalized epistemic and quantum overlaps constructed with respect to RAD have not been normalized since we are primarily concerned with the gap between them, and hence they can achieve a maximum value $\omega_{\textrm{max}}$ which is dependent on $\alpha$ and $\beta$. This leads us to our second observation,

\begin{obs} \label{obs_2}
    The permissible range of generalized epistemic and quantum overlap is,
    \begin{equation}\label{range}
    0\leq \omega_{\Lambda}\leq\omega_Q\leq 1-2\alpha-\beta+\min(1+\beta,2(\alpha+\beta)),
\end{equation}
where $\omega_\Lambda=\omega_Q=0$ when $\bra{\psi_1}\psi_2\rangle=0$ and $\omega_\Lambda=\omega_Q=\omega_{\textrm{max}}=1-2\alpha-\beta+\min(1+\beta,2(\alpha+\beta))$ when $\ket{\psi_1}=\ket{\psi_2}$ (up to an overall phase).
\end{obs}

  The explicit calculation is again deferred to Appendix \ref{appen_gen_overlap}, however, notice that when $\alpha\leq (1-\beta)/2$, $\omega_{\textrm{max}}=(1+\beta)$, and as $\beta\rightarrow0$, $\omega_{\textrm{max}}\rightarrow1$ which is in accordance with the observation that RAD behaves as a distinguishability game with a multiplicative factor of $(1+\beta)$ in this parameter regime. When $\alpha>(1-\beta)/2$, $\omega_{\textrm{max}}=2(1-\alpha)$ and one cannot make $\omega_{\textrm{max}}\rightarrow1$ by varying $\alpha$ in this parameter regime, which again reinforces the implication of Observation \ref{obs_1} that in this parameter regime, RAD is fundamentally different from distinguishability. 
\section{No-go theorem for maximally $\psi-$epistemic models}
Having discussed the nature of the generalized overlaps, we now present two theorems that allow us to show that no maximally $\psi-$epistemic model can explain the statistics arising out of RAD. The theorem that follows this paragraph, Theorem \ref{thm_2}, shows how one can always lower bound the difference of generalized quantum and epistemic overlaps with an operational quantity. Theorem \ref{result} then shows how one can find single-qubit states for which the said operational quantity is strictly greater than zero, and thus both of the theorems taken together form the proof for refuting maximally $\psi-$epistemic models.

\begin{thm} \label{thm_2}
    The difference of generalized quantum and epistemic overlap of two pure quantum states $\ket{\psi_1}$ and $\ket{\psi_2}$ for every value of $\alpha$ and $\beta$ in the range $[0,1]$ is lower bounded by an operational quantity $B_Q(\psi_1,\psi_2,\rho;\alpha,\beta)$, that is,
    \begin{equation} \label{diff_overlap}
        \omega_Q(\psi_1,\psi_2;\alpha,\beta)-\omega_{\Lambda}(\psi_1,\psi_2;\alpha,\beta)\geq B_Q(\psi_1,\psi_2,\rho;\alpha,\beta),
    \end{equation}
    where $\rho=1/2(\ket{\psi_1}\!\bra{\psi_1}+\ket{\psi_2}\!\bra{\psi_2})$ and,
    \begin{equation} \label{b_label}
        \begin{split}
            B_Q(\psi_1,\psi_2,\rho;\alpha,\beta)= & 2\Big((1-\alpha)D_Q(\psi_1,\psi_2)\\
            &+D_Q(\{\psi_1,(1+\beta)/2\},\{\rho,\alpha+\beta\})
        \\
        &+D_Q(\{\psi_2,(1+\beta)/2\},\{\rho,\alpha+\beta\})\\
        &-S^{RAD}_Q(\psi_1,\psi_2;\alpha,\beta)
       -2\beta-1\Big).
        \end{split}
    \end{equation}
\end{thm}
 Please note that $D_Q(\{\psi_i,(1+\beta)/2\},\{\rho,\alpha+\beta\})$ denotes the distinguishability of $\ket{\psi_i}$ and $\rho$ with weights $(1+\beta)/2$ and $(\alpha+\beta)$ respectively. The proof of this theorem has been presented in detail in Appendix \ref{appen_proof_th2}. The important point to note is that the proof utilizes the fact that, if we have a quantum mixed preparation $\sigma$ where, 
\begin{equation} \label{convexity}
\begin{split}
    &\sigma=\sum_iq_i\ket{\phi_i}\!\bra{\phi_i} \textrm{ } \textrm{ then it implies that,}\\
    &\mu(\lambda|\sigma)=\sum_iq_i\mu(\lambda|\phi_i) \quad \forall \lambda\in \Lambda,
    \end{split}
\end{equation}
which is to say that the \textit{convexity of a mixed preparation is preserved in its epistemic state} \cite{Spekkens_2005,mixed}. This allows us to write $\tilde{\mu}_3(\lambda)=2(\alpha+\beta)\mu(\lambda|\rho)$, which then facilitates the proof. 

 Notice what Theorem \ref{thm_2} implies. It has been shown that the gap between the generalized quantum overlap and the generalized epistemic overlap is \textit{always} greater than or equal to a particular operational quantity, $B_Q(\psi_1,\psi_2,\rho;\alpha,\beta)$. Thus, to rule out maximally $\psi-$epistemic models, there must exist at least one pair of quantum states $\ket{\psi_1}$ and $\ket{\psi_2}$ and some parameter pair $(\alpha,\beta)$ such that $B_Q(\psi_1,\psi_2,\rho;\alpha,\beta)>0$. Only then can we claim that there exists \textit{no} maximally $\psi-$epistemic model capable of reproducing quantum statistics. However, we can go even further. We can find out the absolute maximal quantum value of $B_Q(\psi_1,\psi_2,\rho;\alpha,\beta)$, that is, over every possible pair of quantum states, $\{\ket{\psi_1},\ket{\psi_2}\}$, which allows us to comment on how removed the two overlaps must be, without considering any particular epistemic model. Towards this end, we formulate SDP techniques to obtain dimension-dependent lower-bounds and dimension-independent upper-bounds, and present the results via the following theorem, where, as one can see, we have set $\beta=1$. The reason for setting $\beta=1$ and the techniques employed to find the bounds have been described in the paragraph following the theorem. 

     
    
\begin{thm} \label{result}
     The maximum quantum value of $B_Q(\psi_1,\psi_2,\rho;\alpha,\beta=1)$ for all $\alpha\in(\approx0.49,1]$,
     \begin{equation} \label{B_Q(alpha)}
     \begin{split}
B_Q(\alpha)=\max_{\{\psi_1,\psi_2\}}(B_Q(\psi_1,\psi_2,\rho;\alpha,\beta=1)),
\end{split}
\end{equation}
     is achieved by a pair of qubit states, shown in Figure \ref{fig:sub3}, and 
     $\max_{\{\alpha\}}(B_Q(\alpha))\approxeq0.0639$ is achieved at $\alpha\approxeq 0.7124$ for $\ket{\psi_1}=\ket{0}$ and $\ket{\psi_2}=\cos(\pi/3)\ket{0}+\sin(\pi/3)\ket{1}$ (up to machine precision).

\end{thm}

    Figure \ref{fig:sub4} shows that $B_Q(\alpha)$ is achieved by a pair of pure qubit states, $\ket{\psi_1}=\ket{0}$ and $\ket{\psi_2}=\cos(\tilde{\theta}/2)\ket{0}+\sin(\tilde{\theta}/2)\ket{1}$ where $\tilde{\theta}$ is the relative Bloch angle between $\ket{\psi_1}$ and $\ket{\psi_2}$ and $\theta$ is the scaled Bloch angle, i.e. $\theta=\tilde{\theta}/\pi$. We find that changing $\beta$ does not alter the optimal states, and the maximum value is attained when $\beta=1$. The proof is numerical and utilizes SDP techniques \cite{anubhavBoundingEpistemicOverlap} to obtain dimension-dependent lower-bounds and dimension-independent upper-bounds adapted from \cite{Tavakoli2022informationally,Chaturvedi2021characterising}. The detailed description of the techniques has been deferred to Appendix \ref{appen_proof_th3} for brevity.

\begin{figure}
    \centering
    \begin{subfigure}{.4\textwidth}
  \centering
  \includegraphics[width=1\linewidth]{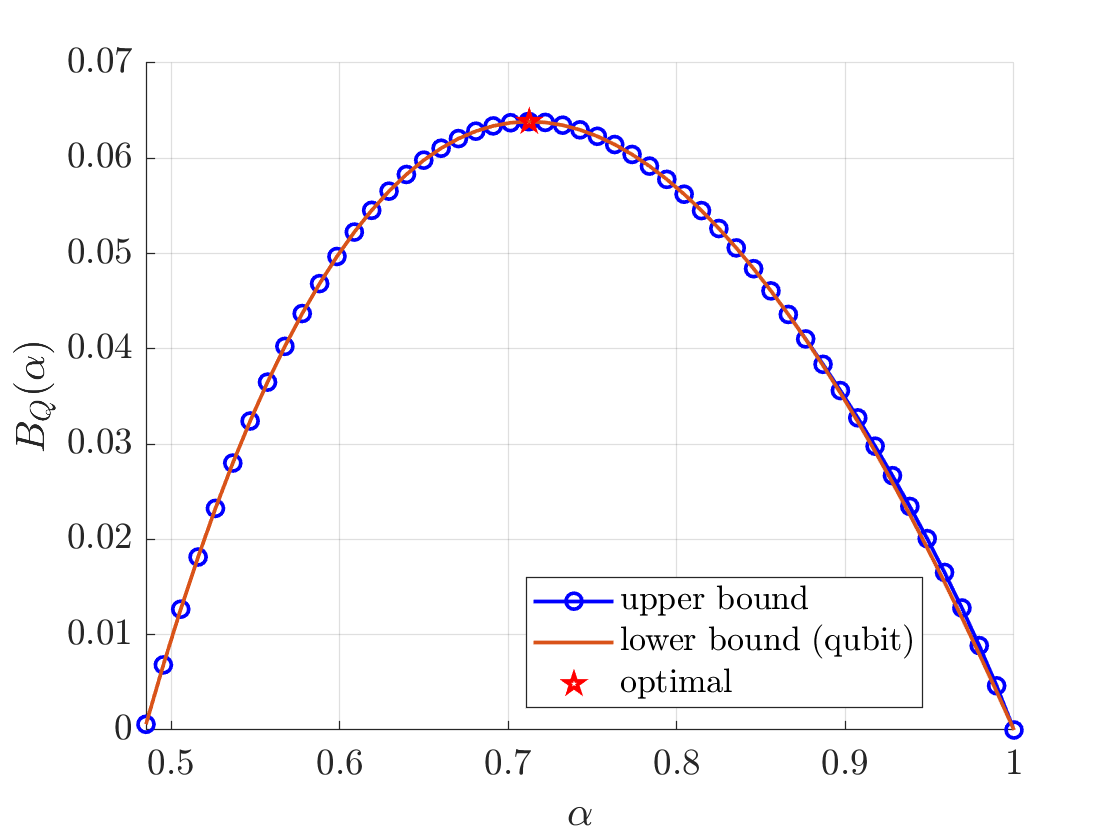}
  \caption{Variation of $B_Q(\alpha)$ with $\alpha$. The red star denotes the maximal value.}
  \label{fig:sub3}
\end{subfigure}
\begin{subfigure}{.4\textwidth}
  \centering
  \includegraphics[width=1\linewidth]{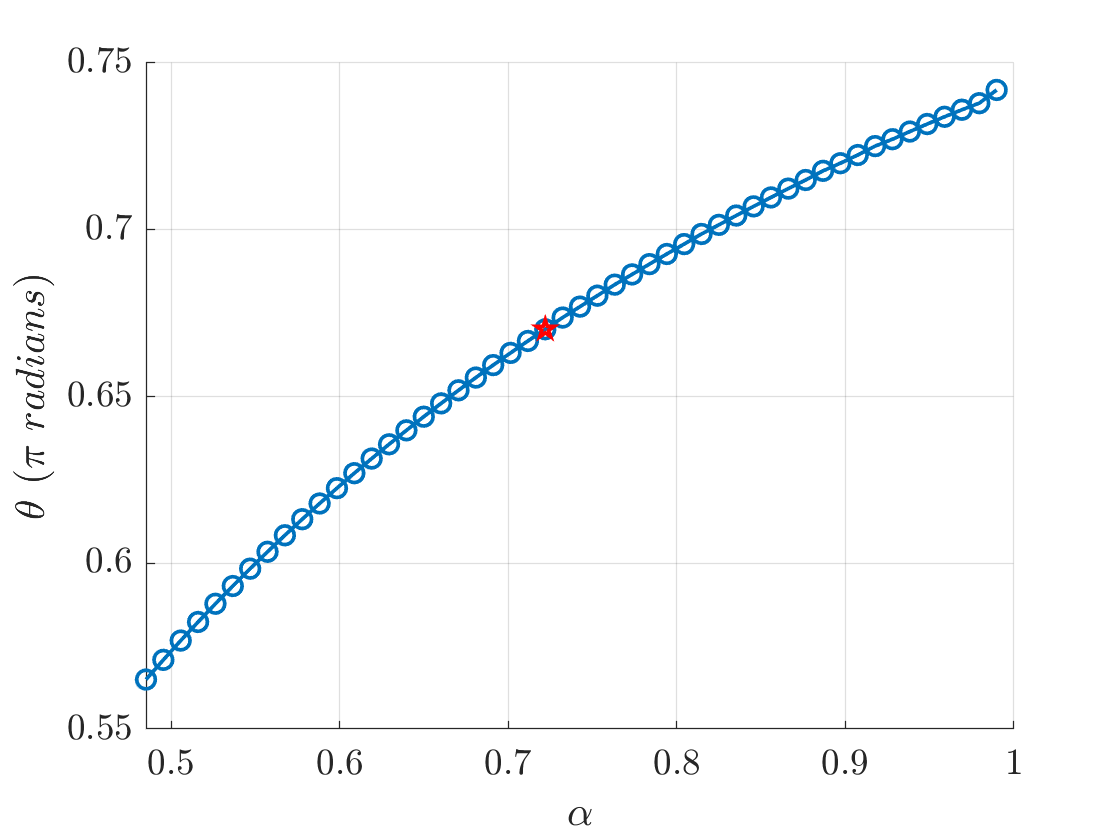}
  \caption{The variation with $\alpha$ of the relative Bloch angle scaled with $1/\pi$ between $\ket{\psi_1}$ and $\ket{\psi_2}$ that results in $B_Q(\psi_1,\psi_2,\rho;\alpha,\beta=1)=B_Q(\alpha)$. The red star denotes the point of maximal value of $B_Q(\alpha)$.}
  \label{fig:sub4}
\end{subfigure}
    \caption{Computational investigation of the gap between the overlaps.}
    \label{fig:op_bound}
\end{figure}
Here, as an example, we present the difference in generalized overlaps using two different epistemic models, namely the Kochen-Specker (KS) model (although it is only defined for projective measurements) and the Renner-Tavakoli-Quintino (RTQ) model \cite{Renner_2023} for optimal $\alpha$ and $\beta$ with optimal qubit states, that is for $\alpha\approxeq0.7124$, $\beta=1$ and $\tilde{\theta}=2\pi/3$,
\begin{equation}
    \begin{split}
        \omega_Q-\omega_{KS}\approxeq 0.1024, \ 
        \omega_Q-\omega_{RTQ}\approxeq 0.1656.
    \end{split}
\end{equation}
As one can see, both differences are well above the operational lower bound set on the difference of overlaps in Theorem \ref{result}.
\section{Advantage in constrained communication scenario}
 One can also see the gap in the overlaps in light of the constrained communication scenarios. In the constrained communication scenarios of \cite{Manna_2024,pandit25}, the success metric of one task is optimized while the messages are constrained by their optimal performance in another task. Such a constraint on the preparations has also been used in semi-device-independent quantum key distribution~\cite{Chaturvedi2026SDIQKD}. The present task is of the same form, with RAD replacing bounded distinguishability or antidistinguishability as the constraint, and with two pure states and their equal mixture as the quantum messages. Consider the following communication task, where Bob (the sender) and Alice (the receiver) have three inputs each, $x\in\{1,2,3\}$ and $y\in\{1,2,3\}$ respectively. For every input, Bob sends a message to Alice, and according to her inputs and the message received from Bob, she produces an output $k$. In effect she tries to distinguish between the messages Bob has sent in the following manner: if $y=1$, she distinguishes between $x=1$ and $x=2$, if $y=2$ she distinguishes between $x=1$ and $x=3$ but with weights $1$ and $1+\alpha$ respectively, and for $y=3$ she distinguishes between $x=2$ and $x=3$ with weights $1$ and $1+\alpha$ respectively. Her goal is to maximize the success metric,
\begin{equation}
\begin{split}
    S^{Comm}=&\frac{1-\alpha}{2}\left(p(1|x=1,y=1)+p(2|x=2,y=1)\right)\\
    &+p(1|x=1,y=2)+(1+\alpha)p(2|x=3,y=2)\\
    &+p(1|x=2,y=3)+(1+\alpha)p(2|x=3,y=3), 
    \end{split}
\end{equation}
subject to the constraint that the optimal average reward in RAD of Bob's messages for inputs $x=1$ and $x=2$ is at most $r$, where $r\in[\alpha,1)$. The measurement used to evaluate this constraint is optimized independently of the measurements used in the communication task. If Bob adopts a classical strategy, then he sends a message $\lambda$ which is chosen from a set $\Lambda$ according to the distributions, $\{\mu(\lambda|x)\}$ such that $\mu(\lambda|x)\geq0$ and $\int_{\Lambda}\mu(\lambda|x)d\lambda=1$ for all $x\in\{1,2,3\}$. This is the \textit{encoding strategy}. The distribution for $x=3$ is an equal mixture of the other two distributions, $\mu(\lambda|x=3)=(1/2)(\mu(\lambda|x=1)+\mu(\lambda|x=2))$. Alice decodes this message and produces her output $k$ by employing the \textit{decoding strategy}, $\{\xi(k|\lambda,y)\}_k$ such that $\xi(k|\lambda,y)\geq0$ and $\sum_k\xi(k|\lambda,y)=1$ for all $\lambda\in\Lambda$ and $y\in\{1,2,3\}$. Then one can write, $p(k|x,y)=\int_{\Lambda}\mu(\lambda|x)\xi(k|\lambda,y)d\lambda$. The classical message, encoding strategy, and decoding strategy have the same mathematical roles as the ontic state, epistemic state, and response scheme, respectively. However, a classical strategy here is not assumed to be maximally $\psi-$epistemic. The optimal classical success metric of the communication task then becomes,
\begin{equation} \label{comm_class}
    \begin{split}
        S^{Comm}_{\Lambda}|_{\textrm{op}}&=\max_{\{\xi(k|\lambda,y)\}_{y}}(S^{Comm})\\
        &=\frac{1-\alpha}{2}\int_{\Lambda}\max\left(\mu(\lambda|x=1),\mu(\lambda|x=2)\right)d\lambda\\
        &+\int_{\Lambda}\max\left(\mu(\lambda|x=1),(1+\alpha)\mu(\lambda|x=3)\right)d\lambda\\
        &+\int_{\Lambda}\max\left(\mu(\lambda|x=2),(1+\alpha)\mu(\lambda|x=3)\right)d\lambda
    \end{split}
\end{equation}
where it is subject to the constraint that $S^{RAD}_{\Lambda}\leq r$, where $S^{RAD}_{\Lambda}$ is as defined in \eqref{def_epis_overlap_RAD} (with $\beta=1$). 

 Using $\mu(\lambda|x=3)=(\mu(\lambda|x=1)+\mu(\lambda|x=2))/2$, the preceding expression can be written as (see Appendix \ref{appen:constrained_com} for details)
\begin{equation} \label{comm_classical}
\begin{split}
S^{Comm}_{\Lambda}|_{\textrm{op}}={}&S^{RAD}_{\Lambda}+3\\
&-\frac{1-\alpha}{2}\int_{\Lambda}\min\left(\mu(\lambda|x=1),\mu(\lambda|x=2)\right)d\lambda.
\end{split}
\end{equation}
A valid strategy in RAD is to distinguish the two messages without using the third outcome, and hence
\begin{equation} \label{RAD_class_bound}
S^{RAD}_{\Lambda}\geq1-\int_{\Lambda}\min\left(\mu(\lambda|x=1),\mu(\lambda|x=2)\right)d\lambda.
\end{equation}
Therefore, every classical strategy satisfying $S^{RAD}_{\Lambda}\leq r$ obeys
\begin{equation} 
S^{Comm}_{\Lambda}|_{\textrm{op}}\leq3+r-\frac{1-\alpha}{2}(1-r).
\label{classical_comm_bound}
\end{equation}

If Bob employs a quantum strategy, then he sends a quantum state, $\rho_x=\ket{\psi_x}\!\bra{\psi_x}$ for $x=1,2$, and $\rho_3=\rho=(\rho_1+\rho_2)/2$ for $x=3$. Accordingly Alice performs the measurements, $\mathcal{M}^{(y)}=\{M_1^{(y)},M_2^{(y)}\}$ such that, $p(k|x,y)=\tr(\rho_xM_k^{(y)})$. The optimum quantum success metric of the communication task then becomes,
\begin{equation}
    \begin{split}
        S^{Comm}_Q|_{\textrm{op}}&=\max_{\{\mathcal{M}^{(y)}\}_y}(S^{Comm})\\
        &=(1-\alpha)D_Q(\psi_1,\psi_2)+D_Q(\{\psi_1,1\},\{\rho,1+\alpha\})\\
    &+D_Q(\{\psi_2,1\},\{\rho,1+\alpha\})
    \end{split}
\end{equation}
subject to the constraint that $S^{RAD}_Q(\psi_1,\psi_2;\alpha,\beta=1)\leq r$. From the definition of $B_Q$ in \eqref{b_label}, for $\beta=1$,
\begin{equation}
B_Q=2\left(S^{Comm}_Q|_{\textrm{op}}-S^{RAD}_Q-3\right).
\end{equation}
We impose the same RAD constraint by choosing $r=S^{RAD}_Q$. Combining the last equation with \eqref{classical_comm_bound} gives, for every classical strategy obeying this constraint,
\begin{equation}
S^{Comm}_Q|_{\textrm{op}}-S^{Comm}_{\Lambda}|_{\textrm{op}}
\geq\frac{1}{2}\left(B_Q+(1-\alpha)(1-r)\right).
\end{equation}
According to Theorem~\ref{result}, for $\alpha\in(\approxeq0.49,1]$ there exist two pure qubit states for which $B_Q>0$. The right-hand side is therefore positive, which establishes a quantum advantage under the same RAD constraint. The relation with the main theorems is direct: Theorem~\ref{thm_2} relates $B_Q$ to the difference of generalized overlaps, and Theorem~\ref{result} shows that $B_Q$ is positive for a pair of pure single-qubit states.

    \section{Discussions and future directions}
     In this article, we presented a no-go theorem for maximally $\psi-$epistemic models for single-qubit systems by providing a non-negative lower bound on the difference of their generalized quantum and epistemic overlap. Here we would like to point out that the choice to use the \textit{difference} of overlaps instead of the \textit{ratio}, as has been done previously in \cite{Barrett,Branciard,Leifer} is a deliberate one. As pointed out by Leifer \cite{ytube,ETHZurich}, to show that $\psi-$epistemic models are unable to describe quantum statistics, one requires  $\omega_{\Lambda}$ to be as small as possible, while at the same time one requires $\omega_Q$ to get as close as possible to $\omega_{\textrm{max}}$. This task is not achievable if one considers the ratio of overlaps $\omega_\Lambda/\omega_Q$ and tries to minimize it, since in every instance that that has been attempted, one ends up with the epistemic overlap tending to zero much faster than the quantum overlap, but the fact of the matter remains that the quantum overlap tends to zero as well, making the set of quantum states considered orthogonal. Thus, a much better way of going about it is to consider the difference of overlaps, where in some asymptotic limit, one must try to show that $(\omega_Q-\omega_{\Lambda})\rightarrow \omega_{\textrm{max}}$. In the future, it would be interesting to see if this can be achieved using RAD by appropriately considering a set of quantum states with increasing dimension and/or number of states. Given the fact that just two pure states of a qubit were sufficient to show a gap, where previously at least a set of $16$ $4-$dimensional quantum states were required to do the same \cite{Barrett}, RAD seems promising. This article also addresses another point raised by Leifer, which is that focusing on the indistinguishability of pure states as a quantifier of their epistemic nature is arbitrary. In fact, not only is focusing on indistinguishability arbitrary, it is also insufficient insofar as it can't constrain $\psi-$epistemic models for a qubit at all. However, this problem goes away when one realises that a better requirement for a model to be maximally $\psi$-epistemic is to demand that the quantum and epistemic overlaps match for \textit{every possible} binary decision task, rather than relying on minimum-error discrimination or distinguishability. This raises the question, what other quantum tasks might be designed for constraining $\psi-$epistemic models? It remains to be seen. 

    \section{Acknowledgements}
     We thank the International Conference on the Foundations of Quantum Mechanics, held in December 2024 at IISER Kolkata, as well as Matthew Leifer and Sahil Gopalkrishna Naik, for the fruitful discussions. DS acknowledges the financial support from STARS (STARS/STARS-2/2023-0809), Govt. of India. A.C. acknowledges support from the KLAR Grant No. BNI/PST/2023/1/00013/U/00001, funded by NAWA.

\appendix

\section{Formulation of universally maximal $\psi$-epistemicity}
\label{sec:general-maximal-psi-epistemicity}

\subsection{Standard overlaps from minimum-error state discrimination}
The conventional notion of epistemic and quantum overlaps is recovered as a special case from \eqref{eq:general-quantum-overlap} and \eqref{eq:general-classical-overlap} when the task considered is minimum-error discrimination or distinguishability.
Consider the two-outcome discrimination task with payoff
\begin{equation}
    g(z,x)=\frac{\delta_{z,x}}{2},
    \qquad z,x\in\{1,2\}.
\end{equation}
The corresponding score is
\begin{equation}
    D_Q(\mathcal{M};\psi_1,\psi_2)
    =
    \frac{p(1|\psi_1,\mathcal{M})+p(2|\psi_2,\mathcal{M})}{2}.
\end{equation}
For two pure states, the Helstrom bound gives
\begin{equation}
    S_Q^{\rm dis}(\psi_1,\psi_2)=D_Q(\psi_1,\psi_2)
    =
    \frac{1+\sqrt{
        1-|\langle\psi_1|\psi_2\rangle|^2
    }}{2}.
    \label{eq:helstrom-score-nohalf}
\end{equation}
Since $S_\perp^{\rm dis}=1$ and $S_=^{\rm dis}=1/2$, the quantum overlap
defined in \eqref{eq:general-quantum-overlap} becomes
\begin{equation}
    \omega_Q^{\rm dis}(\psi_1,\psi_2)
    =
    1-
    \sqrt{
        1-|\langle\psi_1|\psi_2\rangle|^2
    },
    \label{eq:standard-quantum-overlap}
\end{equation}
which is precisely the standard quantum overlap associated with
minimum-error distinguishability.

Likewise, at the ontological level,
\begin{align}
    S_\Lambda^{\rm dis}(\psi_1,\psi_2)
    &=\frac{1}{2}
    \int_\Lambda
    \max\{
        \mu_1(\lambda),
        \mu_2(\lambda)
    \}\,d\lambda
    \nonumber\\
    &=
    1-
    \frac{1}{2}\int_\Lambda
    \min\{
        \mu_1(\lambda),
        \mu_2(\lambda)
    \}\,d\lambda.
\end{align}
Consequently,
\begin{equation}
    \omega_\Lambda^{\rm dis}(\psi_1,\psi_2)
    =
    \int_\Lambda
    \min\{
        \mu_1(\lambda),
        \mu_2(\lambda)
    \}\,d\lambda,
    \label{eq:standard-classical-overlap}
\end{equation}
which is the usual epistemic overlap between the two ontic distributions.
Thus, the conventional criterion of maximal $\psi$-epistemicity based on
state discrimination is precisely the special case
\begin{equation}
    \omega_\Lambda^{\rm dis}(\psi_1,\psi_2)
    =
    \omega_Q^{\rm dis}(\psi_1,\psi_2)
\end{equation}
of \eqref{eq:maximal-psi-epistemic-task}.

\subsection{Transformation monotonicity of the decision score}
Here we illustrate what we mentioned in Section \ref{decision_theoretic}, namely that one cannot apply a transformation on the pair of quantum states to improve the decision score, and more importantly, reduce their epistemic overlap. 
\begin{proposition}[Transformation monotonicity of the decision score] \label{prop:decision-score-monotonicity} Let $g$ be any fixed binary decision task of the form \eqref{eq:general-linear-score}. \begin{enumerate} \item[\rm (i)] \emph{Quantum case.} Let $\{\rho_1,\rho_2\}$ and $\{\sigma_1,\sigma_2\}$ be two pairs of quantum states such that there exists a quantum channel $\mathcal E$ satisfying \begin{equation} \mathcal E(\rho_x)=\sigma_x, \qquad x=1,2. \label{eq:quantum-garbling} \end{equation} Then \begin{equation} \omega_Q^g(\rho_1,\rho_2) \leq \omega_Q^g(\sigma_1,\sigma_2). \label{eq:quantum-score-data-processing} 
\end{equation} 
In particular, for pure states, 
\begin{equation} |\langle\psi_1|\psi_2\rangle| \leq |\langle\phi_1|\phi_2\rangle| \label{eq:pure-overlap-order} \end{equation} 
implies 
\begin{equation} \omega_Q^g(\psi_1,\psi_2) \leq \omega_Q^g(\phi_1,\phi_2). \label{eq:pure-score-order} \end{equation} 

\item[\rm (ii)] \emph{Classical case.} Let $\{\mu_1,\mu_2\}$ and $\{\nu_1,\nu_2\}$ be two pairs of probability distributions over the ontic state space such that the latter pair is obtained from the former by a common stochastic map $K$, namely \begin{equation} \nu_x(\lambda') = \int_\Lambda K(\lambda'|\lambda)\, \mu_x(\lambda)\,d\lambda, \quad x=1,2, \label{eq:classical-garbling} \end{equation} where \begin{equation} K(\lambda'|\lambda)\geq0, \qquad \int d\lambda'\, K(\lambda'|\lambda)=1. \end{equation} Then \begin{equation} \omega_\Lambda^g(\mu_1,\mu_2) \leq \omega_\Lambda^g(\nu_1,\nu_2). \label{eq:classical-score-data-processing} \end{equation} 
Moreover, the stochastic map cannot increase the total-variation distance, \begin{equation} \mathcal{D}_{\rm TV}(\nu_1,\nu_2) \leq \mathcal{D}_{\rm TV}(\mu_1,\mu_2), \label{eq:tv-contraction} \end{equation}
where the total variation distance of two distributions $\{\mu_1,\mu_2\}$ is defined as,
\begin{equation}
    \mathcal{D}_{\rm TV}(\mu_1,\mu_2)=\frac{1}{2}\int_{\Lambda}d\lambda |\mu_1(\lambda)-\mu_2(\lambda)|,
\end{equation}
which is $0$ for identical distributions and $1$ for completely disjoint distributions.
\end{enumerate} \end{proposition} \begin{proof} 
We prove the two statements separately. For the quantum case, let $\{M_z\}_z$ be an arbitrary POVM for the output states $\sigma_x$. Since \begin{equation} \sigma_x=\mathcal E(\rho_x), \end{equation} we have \begin{equation} \tr(M_z\sigma_x) = \tr\!\left[ \mathcal E^\dagger(M_z)\rho_x \right]. \end{equation} Because $\mathcal E$ is trace preserving, its adjoint $\mathcal E^\dagger$ is unital, and therefore \begin{equation} \widetilde M_z := \mathcal E^\dagger(M_z) \end{equation} defines a valid POVM on the input system. Consequently, every set of conditional probabilities obtainable from the pair $\{\sigma_1,\sigma_2\}$ is also obtainable from the pair $\{\rho_1,\rho_2\}$. Hence \begin{equation} S_Q^g(\rho_1,\rho_2) \geq S_Q^g(\sigma_1,\sigma_2), \end{equation} 
and consequently, \eqref{eq:quantum-score-data-processing}.

For the classical case, using \eqref{eq:classical-garbling}, we have \begin{align} S_\Lambda^g(\nu_1,\nu_2) &= \int d\lambda'\, \max_z \left[ g(z,1)\nu_1(\lambda') + g(z,2)\nu_2(\lambda') \right] \nonumber\\ &= \int d\lambda'\, \max_z \int d\lambda\, K(\lambda'|\lambda) \left[ g(z,1)\mu_1(\lambda) + g(z,2)\mu_2(\lambda) \right] \nonumber\\ &\leq \int d\lambda'\int d\lambda\, K(\lambda'|\lambda) \max_z \left[ g(z,1)\mu_1(\lambda) + g(z,2)\mu_2(\lambda) \right] \nonumber\\ &= \int d\lambda\, \max_z \left[ g(z,1)\mu_1(\lambda) + g(z,2)\mu_2(\lambda) \right] \nonumber\\ &= S_\Lambda^g(\mu_1,\mu_2), \end{align} where the inequality follows from the convexity of the maximum and the last equality uses the normalization of the stochastic kernel. Eq.~\eqref{eq:classical-score-data-processing} follows from the above inequality. Finally, \begin{align} \mathcal{D}_{\rm TV}(\nu_1,\nu_2) &= \frac{1}{2} \int d\lambda'\, \left| \int d\lambda\, K(\lambda'|\lambda) \left[ \mu_1(\lambda)-\mu_2(\lambda) \right] \right| \nonumber\\ &\leq \frac{1}{2} \int d\lambda'\int d\lambda\, K(\lambda'|\lambda) \left| \mu_1(\lambda)-\mu_2(\lambda) \right| \nonumber\\ &= \mathcal{D}_{\rm TV}(\mu_1,\mu_2), \end{align} which proves \eqref{eq:tv-contraction}. 
\end{proof}

\section{Proof of Theorem \ref{def_w_lambda}} \label{appen_proof_th1}
 $S^{RAD}_Q(\psi_1,\psi_2;\alpha,\beta)$ expressed in terms of the response schemes and epistemic states (written in terms of the non-normalized distributions defined in \eqref{w_S_mod_rep}), looks like,
 \begin{equation}
\begin{split}
    S^{RAD}_Q=&\frac{1}{2}\max_{\{\mathcal{M}\}}\int_{\Lambda}\Big(\tilde{\mu}_1(\lambda)\xi(1|\lambda,\mathcal{M})+\tilde{\mu}_2(\lambda)\xi(2|\lambda,\mathcal{M})\\
    &+\tilde{\mu}_3(\lambda)\xi(3|\lambda,\mathcal{M})\Big)-\beta.
    \end{split}
\end{equation}
Since the response schemes form a valid set of convex coefficients, we see that,
 \begin{equation} \label{S_lambda}
    \begin{split}
        S^{RAD}_Q\leq S^{RAD}_{\Lambda}=\frac{1}{2}\int_{\Lambda}\max(\tilde{\mu}_1(\lambda),\tilde{\mu}_2(\lambda),\tilde{\mu}_3(\lambda))d\lambda-\beta,
    \end{split}
\end{equation}
following which we use the identity $\max(a,b,c)=a+b+c-\min(a,b)-\min(b,c)-\min(a,c)+\min(a,b,c)$ and the fact that $\int_{\Lambda}\tilde{\mu}_1(\lambda)d\lambda=\int_{\Lambda}\tilde{\mu}_2(\lambda)d\lambda=(1+\beta)$ and $\int_{\Lambda}\tilde{\mu}_3(\lambda)d\lambda=2(\alpha+\beta)$ to express $S_{\Lambda}$ as,
\begin{equation}
    S^{RAD}_{\Lambda}(\psi_1,\psi_2;\alpha,\beta)=1-\frac{\omega_{\Lambda}(\psi_1,\psi_2;\alpha,\beta)}{2},
\end{equation}
where $\omega_{\Lambda}(\psi_1,\psi_2;\alpha,\beta)$ is the generalized epistemic overlap as expressed in Theorem \ref{def_w_lambda}.

\section{Behavior of generalized epistemic overlap and the allowed range of overlaps} \label{appen_gen_overlap}
\subsection{Observation \ref{obs_1}: Behavior of generalized epistemic overlap}
 To understand why the behavior of the generalized epistemic overlap changes across the $\alpha=(1-\beta)/2$ line in the parameter space, notice that the term $T_1$ in \eqref{w_S_mod_rep} is non-zero only over $\Lambda_3$ region and its constribution shows up in both parts of the parameter space as $\int_{\Lambda_3}\min(\tilde{\mu}_1(\lambda),\tilde{\mu}_2(\lambda))d\lambda$. However, the evaluation of the terms $T_2+T_3-T_4$ is non-trivial. Notice that $T_2$ only needs to be evaluated over the region $\Lambda_1+\Lambda_3$, $T_3$ only over $\Lambda_2+\Lambda_3$ and $T_4$ only over $\Lambda_3$. Taking into account the fact that $\forall \lambda \in \Lambda_1$, $\tilde{\mu}_3(\lambda)\leq\tilde{\mu}_1(\lambda)$ and $\forall \lambda \in \Lambda_2$, $\tilde{\mu}_3(\lambda)\leq\tilde{\mu}_2(\lambda)$ we can write,
\begin{equation}
    \begin{split}
        &T_2+T_3-T_4=\int_{\Lambda_1+\Lambda_2}\tilde{\mu}_3(\lambda)d\lambda-2(\alpha+\beta)\\
        &+\int_{\Lambda_3}\left[\min(\tilde{\mu}_1(\lambda),\tilde{\mu}_3(\lambda))+\min(\tilde{\mu}_2(\lambda),\tilde{\mu}_3(\lambda))\right]d\lambda\\
        &-\int_{\Lambda_3}\min(\tilde{\mu}_1(\lambda),\tilde{\mu}_2(\lambda),\tilde{\mu}_3(\lambda))d\lambda.
    \end{split}
\end{equation}
By utilizing the fact that $\int_{\Lambda_1+\Lambda_2}\tilde{\mu}_3(\lambda)d\lambda=2(\alpha+\beta)-\int_{\Lambda_3}\tilde{\mu}_3(\lambda)d\lambda$, which comes from the normalization of $\mu(\lambda|\psi_1)$ and $\mu(\lambda|\psi_2)$, we can write,
\begin{equation}
    \begin{split}
        T_2+T_3-T_4&=\int_{\Lambda_3}\min(\tilde{\mu}_1(\lambda),\tilde{\mu}_3(\lambda))d\lambda\\
        &+\int_{\Lambda_3}\min(\tilde{\mu}_2(\lambda),\tilde{\mu}_3(\lambda))d\lambda\\
        &-\int_{\Lambda_3}\left[\min(\tilde{\mu}_1(\lambda),\tilde{\mu}_2(\lambda),\tilde{\mu}_3(\lambda))+\tilde{\mu}_3(\lambda)\right]d\lambda.
    \end{split}
\end{equation}
The integration over $\Lambda_3$ can be broken apart into integrations over the regions $\Lambda_3'$, $\Lambda_3''$ and $\Lambda_3'''$ such that the above equation reduces to,
\begin{equation}
    \begin{split}
        T_2+T_3-T_4&=\int_{\Lambda_3''}\min(\tilde{\mu}_1(\lambda),\tilde{\mu}_3(\lambda))d\lambda\\
        &+\int_{\Lambda_3''}\min(\tilde{\mu}_2(\lambda),\tilde{\mu}_3(\lambda))d\lambda\\
        &-\int_{\Lambda_3''}\left[\min(\tilde{\mu}_1(\lambda),\tilde{\mu}_2(\lambda),\tilde{\mu}_3(\lambda))+\tilde{\mu}_3(\lambda)\right]d\lambda.
    \end{split}
\end{equation}
Now if $\alpha\leq(1-\beta)/2$, then $\forall \lambda\in \Lambda_3'',$ $\min(\tilde{\mu}_1(\lambda),\tilde{\mu}_3(\lambda))=\min(\tilde{\mu}_2(\lambda),\tilde{\mu}_3(\lambda))=\min(\tilde{\mu}_1(\lambda),\tilde{\mu}_2(\lambda),\tilde{\mu}_3(\lambda))=\tilde{\mu}_3(\lambda)$, which leads to $T_2+T_3-T_4=0$. However, if $\alpha>(1-\beta)/2$, then $\forall \lambda\in \Lambda_3'',$ $\min(\tilde{\mu}_1(\lambda),\tilde{\mu}_3(\lambda))=\tilde{\mu}_1(\lambda)$, $\min(\tilde{\mu}_2(\lambda),\tilde{\mu}_3(\lambda))=\tilde{\mu}_2(\lambda)$ and $\min(\tilde{\mu}_1(\lambda),\tilde{\mu}_2(\lambda),\tilde{\mu}_3(\lambda))=\min(\tilde{\mu}_1(\lambda),\tilde{\mu}_2(\lambda))$, which leads to $T_2+T_3-T_4=\int_{\Lambda_3''}[\max(\tilde{\mu}_1(\lambda),\tilde{\mu}_2(\lambda))-\tilde{\mu}_3(\lambda)]d\lambda$. Adding $T_1$ to $T_2+T_3-T_4$ finally results in equations \eqref{w_S_projective} and \eqref{w_S_POVM}.\\
 \subsection{Observation \ref{obs_2}: Range of generalized overlaps}

 To see how we arrive at the range described in \eqref{range}, notice that both quantum and epistemic overlaps are minimum when both preparations are operationally distinguishable, that is, when $\bra{\psi_1}\psi_2\rangle=0$. This ensures $S^{RAD}_Q|_{\textrm{max}}=S^{RAD}_{\Lambda}|_{\textrm{max}}=1$, which makes $\omega_{\Lambda}|_{\textrm{min}}=\omega_Q|_{\textrm{min}}=0$.

Similarly, both overlaps are maximum when both preparations are completely indistinguishable, that is, when $\ket{\psi_1}=\ket{\psi_2}=\ket{\psi}$ up to an overall phase factor. This implies that for a particular measurement $\mathcal{M}$, the average reward is,
\begin{equation}
    S^{RAD}_{\mathcal{M}}=\frac{1-\beta}{2}+\left((\alpha+\beta)-\left(\frac{1+\beta}{2}\right)\right)p(3|\psi,\mathcal{M}).
\end{equation} 
Thus, for the case when $\alpha\leq (1-\beta)/2$, $(\alpha+\beta)-(1+\beta)/2\leq 0$ and hence $S^{RAD}_{\mathcal{M}}$ is maximized for that measurement for which $p(3|\psi,\mathcal{M})=0$ which implies that $S^{RAD}_Q|_{\textrm{min}}=(1-\beta)/2$ and consequently, $\omega_{Q}|_{\textrm{max}}=(1+\beta)$ whereas if $\alpha>(1-\beta)/2$ then $(\alpha+\beta)-(1+\beta)/2>0$ and $S^{RAD}_{\mathcal{M}}$ is maximized when $p(3|\psi,\mathcal{M})=1$ which implies that $S^{RAD}_Q|_{\textrm{min}}=\alpha$ and consequently $\omega_Q|_{\textrm{max}}=2(1-\alpha)$. Combining these we see that, $\omega_Q|_{\textrm{max}}=1-2\alpha-\beta+\min(1+\beta,2(\alpha+\beta))$. Again, since preparations are pure and indistinguishable, their epistemic states are the same, and from \eqref{w_S_mod_rep}, we see that $\omega_{\Lambda}|_{\textrm{max}}=1-2\alpha-\beta+\min(1+\beta,2(\alpha+\beta))$ as well. This results in \eqref{range}.

\section{Proof of Theorem \ref{thm_2}} \label{appen_proof_th2}
We first present the following lemma.
\begin{lemma} \label{Q_casino_thm}
    Consider any two pure quantum states $\ket{\psi_1}$ and $\ket{\psi_2}$, and their equal mixture $\rho=1/2(\ket{\psi_1}\!\bra{\psi_1}+\ket{\psi_2}\!\bra{\psi_2})$, such that for every value of $\alpha$ and $\beta$ in the range $[0,1]$ the following inequality holds,
    \begin{equation} \label{w_l_ineq}
    \begin{split}
        \frac{1}{2}\omega_{\Lambda}(\psi_1,\psi_2;\alpha,\beta)&\leq2(1+\beta)-(1-\alpha)D_Q(\psi_1,\psi_2)\\
        &-D_Q(\{\psi_1,(1+\beta)/2\},\{\rho,\alpha+\beta\})\\
        &-D_Q(\{\psi_2,(1+\beta)/2\},\{\rho,\alpha+\beta\}).
        \end{split}
    \end{equation}
    
\end{lemma}

Consider the term $T_4$ as defined in \eqref{w_S_mod_rep}. The integrand has the form, $\min(\gamma_1a_1,\gamma_2a_2,\gamma_3a_3)$ where $\gamma_1=\gamma_2=(1+\beta)$ and $\gamma_3=(\alpha+\beta)$ while $a_1=\mu(\lambda|\psi_1)$, $a_2=\mu(\lambda|\psi_2)$ and $a_3=a_1+a_2$. For every $\lambda\in\Lambda$, $a_1,a_2$ and $a_3$ are non-negative numbers and so are $\gamma_1,\gamma_2$ and $\gamma_3$, thus $\min(\gamma_1a_1,\gamma_2a_2,\gamma_3a_3)\geq\min(\gamma_1,\gamma_2,\gamma_3)\min(a_1,a_2,a_3)$ holds. In this context,then, $\min(\gamma_1,\gamma_2,\gamma_3)=(\alpha+\beta)$, and $\min(a_1,a_2,a_3)=\min(a_1,a_2)$ since $a_3=a_1+a_2$ which allows us to write,
\begin{equation}
T_4\geq(\alpha+\beta)\int_{\Lambda}\min(\mu(\lambda|\psi_1),\mu(\lambda|\psi_2))d\lambda.
\end{equation}
Since the \textit{convexity of a mixed preparation is preserved in its epistemic state}, we write $\tilde{\mu}_3(\lambda)=2(\alpha+\beta)\mu(\lambda|\rho)$ where $\rho=1/2(\ket{\psi_1}\!\bra{\psi_1}+\ket{\psi_2}\!\bra{\psi_2})$ and use the lower bound on $T_4$ as expressed above to present the following bound on $\omega_{\Lambda}(\psi_1,\psi_2;\alpha,\beta)$,
\begin{equation} \label{w_upper_epis}
    \begin{split}
        \omega_{\Lambda}\leq &(1-\alpha)\int_{\Lambda}\min(\mu(\lambda|\psi_1),\mu(\lambda|\psi_2))d\lambda\\
        &+2\int_{\Lambda}\min\left(\frac{(1+\beta)}{2}\mu(\lambda|\psi_1),(\alpha+\beta)\mu(\lambda|\rho)\right)d\lambda\\
        &+2\int_{\Lambda}\min\left(\frac{(1+\beta)}{2}\mu(\lambda|\psi_2),(\alpha+\beta)\mu(\lambda|\rho)\right)d\lambda\\
        &-2(\alpha+\beta).
    \end{split}
\end{equation}
The terms on the right-hand side of the above inequality can be replaced by operational quantities. That is so because, if one considers weighted distinguishability as defined in the letter and expresses it in terms of epistemic states and response schemes, much like what was done with $S^{RAD}_Q$ in the proof of Theorem \ref{def_w_lambda}, then it can be shown that,
\begin{equation}
\begin{split}
    &D_Q(\{\phi_1,w_1\},\{\phi_2,w_2\})\leq w_1+w_2\\
    &-\int_{\Lambda}\min(w_1\mu(\lambda|\phi_1),w_2\mu(\lambda|\phi_2))d\lambda.
    \end{split}
\end{equation}
Using this inequality, one can replace the right-hand side of \eqref{w_upper_epis} with a linear combination of $D_Q(\psi_1,\psi_2)$, $D_Q(\{\psi_1,(1+\beta)/2\},\{\rho,(\alpha+\beta)\})$ and $D_Q(\{\psi_2,(1+\beta)/2\},\{\rho,(\alpha+\beta)\})$ which results in \eqref{w_l_ineq}. Once this inequality is established, it is quite straightforward to arrive at \eqref{diff_overlap} with $B_Q(\psi_1,\psi_2,\rho; \alpha, \beta)$ as defined in \eqref{b_label}. 
\section{SDP formulation of Risk Aware Distinguishability and proof of Theorem \ref{result}} \label{appen_proof_th3}
The primal of the RAD task is formulated as follows:
 \begin{equation} \label{primal}
 \begin{aligned}
S^{RAD}_Q=\max_{\{M_1,M_2,M_3\}}\quad & \frac{1}{2}\sum_{i=1}^3 \tr( C_i M_i) \\
\text{subject to}\quad & \sum_{i=1}^3 M_i = \mathbbm{1},\\
& M_i \succeq 0,\quad i=1,2,3,
\end{aligned}
 \end{equation}
where $C_1=\ket{\psi_1}\!\bra{\psi_1}-\beta\ket{\psi_2}\!\bra{\psi_2}$, $C_2=\ket{\psi_2}\!\bra{\psi_2}-\beta\ket{\psi_1}\!\bra{\psi_1}$ and $C_3=\alpha(\ket{\psi_1}\!\bra{\psi_1}+\ket{\psi_2}\!\bra{\psi_2})$.
The dual of this problem is formulated as
\begin{equation} \label{dual}
    \begin{aligned}
\text{min}\quad &\tr(Y)\\
\text{subject to}\quad & Y \succeq \tfrac{1}{2} C_i,\quad i=1,2,3.
\end{aligned}
\end{equation}
Since the Slater's condition is satisfied, we always find a $Y=Y^*\succeq0$ such that $S^{RAD}_Q=\tr(Y^*)$.\\

 Having described the SDP formulation of RAD, we now focus on the numerical SDPs programming techniques \cite{anubhavBoundingEpistemicOverlap} used to find the dimension-independent maximum quantum value of $B_Q(\psi_1,\psi_2,\rho;\alpha,\beta)$ which we denote with $B_Q(\alpha,\beta)$. Note that in Theorem \ref{result}, we have presented maximum quantum value of $B_Q(\psi_1,\psi_2,\rho;\alpha,\beta=1)$ denoted as $B_Q(\alpha)$ in \eqref{B_Q(alpha)}.
\subsection{See-saw method for dimension-dependent lower bounds}
 Recall the expression of $B_Q(\psi_1,\psi_2,\rho;\alpha,\beta)$ as given by \eqref{b_label}. Notice that it can be expressed as a difference of two quantities, $B_Q(\psi_1,\psi_2,\rho;\alpha,\beta)=2(S^{Comm}_Q|_{\textrm{op}}-T)$ where $T=\big(S_Q^{RAD}+2\beta+1\big)$ and $S^{Comm}_Q|_{\textrm{op}}$ is the optimal success metric of the communication task, almost similar to the one described in the letter, except that here we have not set $\beta=1$ yet, that is,
 \begin{equation} \label{op_comm_task_beta}
\begin{split}
    S_Q^{Comm}\big|_{\textrm{op}}=&(1-\alpha)D_Q(\psi_1,\psi_2)\\
    &+D_Q(\{\psi_1,(1+\beta)/2\},\{\rho,\beta+\alpha\})\\
    &+D_Q(\{\psi_2,(1+\beta)/2\},\{\rho,\beta+\alpha\}).
    \end{split}
\end{equation}
 
 So, to find the maximum quantum value of $B_Q(\psi_1,\psi_2,\rho;\alpha,\beta)$, we need to maximize $S^{Comm}_Q|_{\textrm{op}}$ and simultaneously minimize $S_Q^{RAD}$. Since \(S^{RAD}_Q\) is itself a maximum over RAD strategies, the problem has a max–min structure. 

Let us consider the first term. Given two states (not necessarily pure) $\rho_1$ and $\rho_2$ and their equal mixture $\rho=(\rho_1+\rho_2)/2$, retrieving the maximum value of the first term forms a straightforward SDP problem,
\begin{equation} \label{primalD}
 \begin{aligned}
S^{Comm}_Q(\rho_1,\rho_2)|_{\textrm{op}}=\max_{\{M^y_i\}_{i,y}}\quad & \sum_{y=1}^3\sum_{i=1}^21 \tr( K^y_i M^y_i) \\
\text{subject to}\quad & \sum_{i=1}^2 M^y_i = \mathbbm{1}, \quad y=1,2,3,\\
& M^y_i \succeq 0,\quad i=1,2,\\ &y=1,2,3,
\end{aligned}
 \end{equation}
where $K_1^1=(1-\alpha)\rho_1/2$, $K_2^1=(1-\alpha)\rho_2/2$, $K_1^2=\rho_1(1+\beta)/2$, $K_2^2=(1/2)(\beta+\alpha)(\rho_1+\rho_2)$, $K_1^3=\rho_2(1+\beta)/2$ and finally $K_2^3=K_2^2$. For the second term, we follow the approach of \cite{Tavakoli2022informationally}, and consider the dual of RAD \eqref{dual}, so that our problem of finding $B_{Q}(\alpha,\beta)$ becomes
\begin{equation} \label{DualD}
 \begin{aligned}
&B_{Q}(\alpha,\beta)=\max_{\{\rho_1,\rho_2,Y\}}2\left( S^{Comm}_Q(\rho_1,\rho_2)|_{\textrm{op}}-\tr(Y)-2\beta-1\right) \\
&\text{subject to}\quad   Y \succeq \tfrac{1}{2} C_i,\quad i=1,2,3,
\end{aligned}
 \end{equation}
where $C_1=\rho_1-\beta\rho_2$, $C_2=\rho_2-\beta\rho_1$, and $C_3=\alpha(\rho_1+\rho_2)$.

For any given Hilbert space dimension, \eqref{primalD} and \eqref{DualD} give a direct way to lower bound $B_Q(\alpha,\beta)$ via an alternating SDP. Specifically, starting out with random states (mixed in general) $\rho_1$ and $\rho_2$, we find optimal POVMs by solving \eqref{primalD}, then fixing the optimal POVMs, we solve \eqref{DualD}, retrieving optimal $\rho_1,\rho_2$, and $Y$, and supplying them back as constants to \eqref{primalD}. The program iterates untill convergence and yields a lower bound 
\begin{equation} \label{seesawLB}
    B_{Q_{L}}(\alpha,\beta)\leq B_{Q}(\alpha,\beta).
\end{equation}
Moreover, the program also returns the optimal states $\rho_1,\rho_2$ and measurements $\{M^y_i\}$. We find that for all values of $\alpha,\beta$, qubits reach the maximum gap and the states $\rho_1,\rho_2$ are pure, that is, they satisfy $\rho_1^2=\rho_1,\rho_2^2=\rho_2$ (up to machine precision). We find that changing $\beta$ does not alter the optimal states and $\beta=1$ provides the maximum $B_{Q_{L}}(\alpha,\beta)$. Hence, we plot $B_{Q_{L}}(\alpha)=B_{Q_{L}}(\alpha,\beta=1)$ and the optimal Bloch angle $\theta$ between the optimal states against $\alpha$ in FIG. \ref{fig:op_bound}. Finally, we find that absolute maximum gap $B_{Q_{L}}(\alpha)=0.0639$ is attained for $\alpha=0.7124$
However, this approach yields only dimension-dependent lower bounds on $B_Q(\alpha,\beta)$. 
\subsection{Dimension independent upper-bounds via tracial non-commuting polynomial optimization}
To find dimension-independent upper bounds, we formulate a tracial non-commuting polynomial optimization problem following \cite{Tavakoli2022informationally,Chaturvedi2021characterising}. First, following \cite{Chaturvedi2021characterising}, we consider $3$ moment matrices $\Gamma_{\rho_1},\Gamma_{\rho_2},\Gamma_{Y}$, hinged on semidefinite variables $\rho_1,\rho_2,Y$, with a common level-$2$ operator list $\mathcal{O}\equiv\{\mathbbm{1},\{M^y_1\}_{y\in\{1,2,3\}},\{M^y_1M^{y'}_1\}_{y\neq y'\in\{1,2,3\}}\}$, such that their entries are defined as, 
\begin{equation}
[\Gamma_{X}]_{O_i,O_j}=\tr(XO^\dagger_iO_j)    
\end{equation}
for all $X\in\{\rho_1,\rho_2,Y\}$ and $O_i,O_j\in\mathcal{O}$. We then impose normalization constraints, 
\begin{equation} \label{normCon}
[\Gamma_X]_{\mathbbm{1},\mathbbm{1}}=1,     
\end{equation}
for all $X\in\{\rho_i,\rho_j\}$, and constraints encoding the projectivity of the measurement operators including, 
\begin{equation}\label{projCon}
[\Gamma_{X}]_{\mathbbm{1},M^y_1}=[\Gamma_{X}]_{M^y_1,M^y_1},
\end{equation}
for all $y\in\{1,2,3\}$ $X\in\{\rho_1,\rho_2,Y\}$. Additionally, following \cite{Tavakoli2022informationally} we impose the following constraints implied by semidefinite inequality constraints in \eqref{DualD}, specifically,
\begin{equation}
    \Gamma_{Y}\succeq \frac{1}{2}C''_i,
\end{equation}
for all $i\in\{1,2,3\}$, where $C''_1=\Gamma_{\rho_1}-\beta\Gamma_{\rho_2}$, $C''_2=\Gamma_{\rho_2}-\beta\Gamma_{\rho_1}$ and $C''_3=\alpha(\Gamma_{\rho_1}+\Gamma_{\rho_2})$.

To impose the pure state constraints $\rho^2_1=\rho_1$ and $\rho^2_2=\rho_2$, we additionally include a separate moment matrix $\Gamma$ with an operator list
\begin{equation}
\begin{split}
\mathcal{O}_{\Gamma}\equiv&\{\mathbbm{1},\rho_1,\rho_2,\{M^y_1\}_{y\in\{1,2,3\}},\{\rho_i\rho_j\}_{i,j\in\{1,2\}},\\
&\{\rho_iM^y_1\}_{i\in\{1,2\},y\in\{1,2,3\}}\}
\end{split}
\end{equation}

with entries $[\Gamma]_{O_i,O_j}=\tr(O^\dagger_iO_j)$ for all $O_i,O_j\in \mathcal{O}_{\Gamma}$, and impose the constraints, 
\begin{equation}\label{pureCon}
    [\Gamma]_{\mathbbm{1},\rho_{i}}=[\Gamma]_{\rho_{i},\rho_{i}}=1,
\end{equation}
for all $i\in\{1,2\}$. We further impose additional constraints interconnecting the entries among the moment matrices.
Finally, we maximize the following (linear on the entries of moment matrices) objective function,
\begin{align} \nonumber
    &(1+\beta)[\Gamma_{\rho_1}]_{\mathbbm{1},M^2_1}+2(\alpha+\beta)\left(1-\frac{[\Gamma_{\rho_1}]_{\mathbbm{1},M^2_1}+[\Gamma_{\rho_2}]_{\mathbbm{1},M^2_1}}{2}\right)+ \\ \nonumber
    &(1+\beta)[\Gamma_{\rho_2}]_{\mathbbm{1},M^3_1}+2(\alpha+\beta)\left(1-\frac{[\Gamma_{\rho_1}]_{\mathbbm{1},M^3_1}+[\Gamma_{\rho_2}]_{\mathbb{I},M^3_1}}{2}\right)+\\ 
    &(1-\alpha)\big([\Gamma_{\rho_1}]_{\mathbbm{1},M^1_1}-[\Gamma_{\rho_2}]_{\mathbbm{1},M^1_1}+1\big)-2([\Gamma_{Y}]_{\mathbbm{1},\mathbbm{1}}+2\beta+1)
\end{align}
Finally, imposing the semidefinite positivity of the moment matrices, $\Gamma_{\rho_1}\succeq 0,\Gamma_{\rho_2}\succeq 0,\Gamma_{Y}\succeq 0,\Gamma\succeq 0$, the program returns dimension independent upper-bounds $B_{Q_{UB}}(\alpha,\beta)$, such that, 
\begin{equation} \label{ncpopUB}
B_{Q}(\alpha,\beta)\leq B_{Q_{UB}}(\alpha,\beta).
\end{equation}
We plot the upper bounds $B_{Q_{UB}}(\alpha)=B_{Q_{UB}}(\alpha,\beta=1)$ in FIG. \ref{fig:op_bound}. Whenever the upper bounds from the tracial non-commuting polynomial optimization \eqref{ncpopUB} match the lower bounds from the seesaw method \eqref{seesawLB}, we retrieve the maximum quantum value such that $B_{Q_{UB}}(\alpha)=B_{Q_{L}}(\alpha)=B_{Q}(\alpha)$. In particular, for all values of $\alpha\in(\approx0.49,1)$ $B_{Q_{UB}}(\alpha)=B_{Q_{L}}(\alpha)=B_{Q}(\alpha)$ (up to machine precision).

\section{Intermediate steps relating constrained communication success metric and average reward in RAD} \label{appen:constrained_com}
We provide the intermediate steps leading to equations \eqref{comm_classical} and \eqref{RAD_class_bound} in this section of the Appendix.\\

Using the fact that,
\begin{equation} \label{eq_mix}
\mu(\lambda|x=3)
=
\frac{\mu(\lambda|x=1)+\mu(\lambda|x=2)}{2},
\end{equation}
we first note that, for every $\lambda\in\Lambda$,
\begin{equation}
\begin{split}
(1+\alpha)\mu(\lambda|x=3)
&=
\frac{1+\alpha}{2}
\left(
\mu(\lambda|x=1)+\mu(\lambda|x=2)
\right)\\
&\geq
\min\left(
\mu(\lambda|x=1),\mu(\lambda|x=2)
\right),
\end{split}
\end{equation}
where we have used $\alpha\in[0,1]$. Consequently,
\begin{equation} \label{eq:max_identity_comm}
\begin{split}
&\max\left(
\mu(\lambda|x=1),(1+\alpha)\mu(\lambda|x=3)
\right)\\
&+
\max\left(
\mu(\lambda|x=2),(1+\alpha)\mu(\lambda|x=3)
\right)
\\
&=
\max\left(
\mu(\lambda|x=1),
\mu(\lambda|x=2),
(1+\alpha)\mu(\lambda|x=3)
\right)\\
&+(1+\alpha)\mu(\lambda|x=3).
\end{split}
\end{equation}
Moreover,
\begin{equation} \label{max_min}
\begin{split}
&\max\left(
\mu(\lambda|x=1),\mu(\lambda|x=2)
\right)\\
&=
\mu(\lambda|x=1)+\mu(\lambda|x=2)
-
\min\left(
\mu(\lambda|x=1),\mu(\lambda|x=2)
\right).
\end{split}
\end{equation}
Substituting these identities into \eqref{comm_class}, and again using \eqref{eq_mix}, gives
\begin{equation} \label{eq:comm_intermediate}
\begin{split}
S_{\Lambda}^{\mathrm{Comm}}\big|_{\mathrm{op}}
=&
\int_{\Lambda}
\Bigg[
\mu(\lambda|x=1)+\mu(\lambda|x=2)\\
&+
\max\left(
\mu(\lambda|x=1),
\mu(\lambda|x=2),
(1+\alpha)\mu(\lambda|x=3)
\right)\\
&-
\frac{1-\alpha}{2}
\min\left(
\mu(\lambda|x=1),\mu(\lambda|x=2)
\right)
\Bigg]d\lambda.
\end{split}
\end{equation}
Since the encoding distributions are normalized,
\[
\int_{\Lambda}\mu(\lambda|x=1)d\lambda
=
\int_{\Lambda}\mu(\lambda|x=2)d\lambda
=1
\]
and furthermore, for $\beta=1$, the optimal classical reward in RAD from \eqref{S_lambda} is
\begin{equation} \label{appen_RAD_class}
\begin{split}
&S_{\Lambda}^{\mathrm{RAD}}+1\\
&=
\int_{\Lambda}
\max\left(
\mu(\lambda|x=1),
\mu(\lambda|x=2),
(1+\alpha)\mu(\lambda|x=3)
\right)d\lambda
\end{split}
\end{equation}
which leads to \eqref{comm_classical}. 

To arrive at \eqref{RAD_class_bound}, notice that \eqref{appen_RAD_class} implies that 
\begin{equation}
    S^{\textrm{RAD}}_\Lambda+1\geq \int_{\Lambda}\max(\mu(\lambda|x=1),\mu(\lambda|x=2))d\lambda.
\end{equation}
To the above equation one applies \eqref{max_min} and the normalizations of $\{\mu(\lambda|x=1)\}$ and $\{\mu(\lambda|x=2)\}$ which leads to \eqref{RAD_class_bound}.

\end{document}